\documentclass[runningheads,a4paper]{llncs}

\usepackage{amssymb}
\usepackage{graphicx}

\usepackage{url}
\urldef{\mails}\path|{igor.bjelakovic, boche, gisbert.janssen, janis.noetzel}@tum.de|

\usepackage{amsfonts}
\usepackage{amsmath}
\usepackage{dcolumn}
\usepackage{bm}
\usepackage{bbm}
\usepackage{verbatim}
\usepackage{color}
\usepackage{marvosym}
\newcommand{\hr}{{\mathcal H}}
\newcommand{\cn}{{\mathcal N }}

\newcommand{\cs}{{\mathcal S}}
\newcommand{\crr}{{\mathcal R}}
\newcommand{\fr}{{\mathcal F}}

\newcommand{\fri}{{\mathfrak I}}
\newcommand{\kr}{{\mathcal K}}
\newcommand{\bo}{{\mathcal B}}
\newcommand{\cc}{{\mathbb C}}
\newcommand{\rr}{{\mathbb R}}
\newcommand{\M}{{\mathcal M}}
\newcommand{\nn}{{\mathbb N}}

\newcommand{\eps}{{\varepsilon}}

\newcommand{\A}{\mathcal A}
\newcommand{\B}{\mathcal B}
\newcommand{\D}{\mathcal D}
\newcommand{\cP}{\mathcal P}
\newcommand{\bS}{\mathbf S}
\newcommand{\bX}{\mathbf X}
\newcommand{\W}{\mathcal W}
\newcommand{\bA}{\mathbf A}
\newcommand{\bB}{\mathbf B}

\newcommand{\tr}{\mathrm{tr}}
\newcommand{\supp}{\mathrm{supp}}

\DeclareMathOperator{\conv}{conv}

\newcommand{\bra}[1]{{\langle #1 |}}
\newcommand{\ket}[1]{{| #1 \rangle}}
\newcommand{\eins}{{\mathbbm{1}}} 
\newcommand{\erw}{{\mathbbm{E}}}
\newcommand{\pr}{{\mathfrak{P}}}
\DeclareMathOperator{\ran}{ran}
\DeclareMathOperator{\spec}{spec}
\DeclareMathOperator{\argmax}{argmax}

\begin{document}

\mainmatter  

\title{Arbitrarily varying and compound classical-quantum channels and a note on quantum zero-error capacities}

\titlerunning{Arbitrarily varying \& compound classical-quantum channels...}

%
%
\author{Igor Bjelakovi\'{c}\textsuperscript{1} \and Holger Boche\textsuperscript{2} \and Gisbert Jan\ss en\textsuperscript{1} \and Janis N\"otzel\textsuperscript{1}}
\authorrunning{Arbitrarily varying \& compound classical-quantum channels...}

\institute{\textsuperscript{1}Theoretische Informationstechnik, Technische Universit\"at M\"unchen\\
80290 M\"unchen, Germany\\
\textsuperscript{2}Lehrstuhl f\"ur Theoretische Informationstechnik, Technische Universit\"at M\"unchen\\
80290 M\"unchen, Germany\\
\vspace{2ex}
\mails\\
}

%
%

\maketitle

\begin{center}\small\textit{
  Dedicated to the memory of Rudolf Ahlswede}
\end{center}

\begin{abstract}
We consider compound as well as arbitrarily varying classical-quantum channel models. For classical-quantum compound channels, we give an elementary proof of the direct part of the coding theorem.
A weak converse under average error criterion to this statement is also established. 
We use this result together with the robustification 
and elimination technique developed by Ahlswede in order to give an alternative proof of the direct part of the coding theorem for a 
finite classical-quantum arbitrarily 
varying channels with the criterion of success being average error probability. Moreover we provide a proof of the strong converse to the random 
coding capacity in this setting.\\
The notion of symmetrizability for the maximal error probability is defined and it is shown to be both necessary and sufficient for the capacity for
 message transmission with maximal error probability criterion to equal zero.\\
Finally, it is shown that the connection between zero-error capacity and certain arbitrarily varying channels is, just like in the case of 
quantum channels, only partially valid for classical-quantum channels.
\end{abstract}

\begin{section}{Introduction\label{introduction}}
Channel uncertainty is omnipresent and mostly unavoidable in real-world applications and one of the major technological 
challenges is the design of communication protocols that are robust against it. The incarnation of that challenge on the theoretical side
delivers a plethora of interesting structural and methodological problems for Information Theory. Despite these facts it happened only recently
 that this range of problems received the necessary attention in Quantum Information Theory and especially in Quantum Shannon Theory 
\cite{ahlswede-blinovsky}, \cite{datta07}, \cite{bb-compound}, \cite{bbn-2}, \cite{abbn}. 
In this paper we revisit two basic models for communication
under channel uncertainty, the compound and arbitrarily varying channels with classical input and quantum output and give essentially  
self-contained derivations of coding theorems for them. These results were originally obtained in \cite{ahlswede-blinovsky} and \cite{bb-compound}.\\
The contributions of the paper and the difference to existing work are the following. First, in \cite{bb-compound} a capacity result with strong converse for compound channels with a 
classical input and quantum output (compound cq-channel for short)
under the maximum error criterion has been derived. However, the achievability proof given there lacks transparency and does not show that 
good codes with the
 uniformly bounded exponentially decreasing maximal error exist. Indeed, in \cite{bb-compound}  it is merely shown that good codes exist with uniformly 
super-polynomially decreasing maximal error probability. 
Here we prove that sharper result for the average error criterion and, at the same time, give a significantly simpler proof of 
the achievability part of the coding theorem based on a universal
 hypothesis testing result which is a generalization of the technique developed by Hayashi and Ogawa in \cite{ogawa01}. The passage to the maximal 
error criterion can be carried out via a standard argument which can be found in \cite{bb-compound}. \\
It is interesting to compare this result with related work of Hayashi  \cite{hayashi08} and Datta and Hsieh \cite{datta-hsieh}.
The works \cite{hayashi08} and \cite{datta-hsieh} aim at showing the existence of codes depending on the input distribution and a prescribed
rate only and achieving an exponential but \emph{channel dependent} decay of error probability for all cq-channels whose Holevo information is strictly
larger than that prescribed rate. The good codes in our approach depend on 
the input distribution and the set of cq-channels generating the compound cq-channel. Additionally we obtain a uniform exponential 
bound on error probabilities, a property that seems highly desirable in case that the channel is unknown.\\ 
Moreover, 
we prove the weak converse to the coding theorem under average error criterion by a reduction to the strong converse for the  maximal error via a
 lemma of Ahlswede and Wolfowitz from \cite{ahlswede69}.\\
Second, once we have the achievability result for  compound cq-channels we can obtain the corresponding results for arbitrarily varying cq-channels
 (AVcqC) in a straight-forward
fashion via Ahlswede's powerful elimination \cite{ahlswede-elimination} and robustification \cite{ahlswede-coloring} techniques. This way, we obtain an 
alternative approach to the coding theorem for AVcqCs which was originally  proven by Ahlswede and Blinovsky in \cite{ahlswede-blinovsky}.\\
Finally, we show that a naive quantum analog of Ahlswede's beautiful relation \cite{ahlswede-note} between Shannon's zero-error capacity \cite{shannon}
 and
 the capacity of arbitrarily varying channels subject to maximal error criterion does hold neither for AVcqCs when employing the maximal error
 criterion  nor for the strong subspace transmission  over arbitrarily varying quantum channels. The latter communication scenario is widely
 acknowledged as a fully quantum counterpart to message transmission subject to the maximal error criterion. 
\end{section}

\begin{section}{\label{sec:Notation}Notation and Conventions}
All Hilbert spaces are assumed to have finite dimension and are over the field $\cc$. The set of linear operators from $\hr$ to $\hr$ is denoted $\mathcal B(\hr)$. The adjoint 
of $b\in\mathcal B(\hr)$ is marked by a star and written $b^\ast$. The notation $\langle\cdot,\cdot\rangle_{HS}$ is reserved for the Hilbert-Schmidt inner product on $\mathcal B(\hr)$.\\
$\mathcal{S}(\hr)$ is the set of states, i.e. positive semi-definite operators with trace $1$ acting on the Hilbert space $\hr$. Pure states are given by projections onto one-dimensional subspaces. A vector $x\in\hr$ of unit length spanning such a subspace will therefore be referred to as a state vector, the corresponding state will be written 
$|x\rangle\langle x|$. For a finite set $\mathbf X$ the notation $\mathfrak{P}(\mathbf X)$ is reserved for the set of probability distributions on $\mathbf X$, and 
$|\mathbf X|$ denotes its cardinality. For any $l\in\nn$, we define $\bX^l:=\{(x_1,\ldots,x_l):x_i\in\bS\ \forall i\in\{1,\ldots,l\}\}$, we also write $x^l$ for the elements of $\bX^l$.
For any natural number $N$, we define $[N]$ to be the shortcut for the set $\{1,...,N\}$\\
The set of classical-quantum channels (cq-channels) mapping a finite alphabet $\mathbf X$ to a Hilbert space $\hr$ is denoted $CQ(\mathbf{X},\hr)$. Since $CQ(\mathbf X,\hr)$ 
is the set of functions $W:\mathbf X \rightarrow \cs(\hr)$. It is naturally equipped with the norm $\|\cdot\|_{cq}$ (which is inherited from the usual one-norm $\|\cdot\|_1$ 
on operators) and is defined by
\begin{align*}
\|W\|_{cq} := \max_{x \in \mathbf{X}} \|W(x)\|_1 &&(W \in CQ(\mathbf{X},\hr)).
\end{align*}
It is common, to embed the set $\pr(\mathbf{X})$ of probability distributions into $\bo(\cc^{|\mathbf{X}|})$, i.e. to fix an 
orthonormal basis $\{e_x\}_{x \in \mathbf{X}}$ in $\cc^{|\mathbf{X}|}$ and assign to every $p$ in $\pr(\mathbf{X})$ an 
element of $\bo(\cc^{|\mathbf{X}|})$ which is diagonal in this 
basis. For a channel $W\in CQ(\mathbf{X},\hr)$ and a given input probability distribution $p \in \pr(\mathbf{X})$ one defines the corresponding state on $\cc^{|\mathbf{X}|} \otimes \hr$ by
\begin{align}
 \rho := \sum_{x \in \mathbf{X}} p(x) \ket{e_x}\bra{e_x} \otimes W(x).
\end{align}
The set of measurements with $N\in\nn$ different outcomes is written  
$\M_N(\hr):=\{(D_1,\ldots,D_N):\sum_{i=1}^ND_i\leq\eins_\hr\ \mathrm{and}\ D_i\geq 0 \ \forall i\in[N]\}$. 
To every $(D_1,\ldots,D_N)\in \M_N(\hr)$ there corresponds a unique operator defined by $D_0:=\eins_\hr-\sum_{i=1}^ND_i $.\\

The von Neumann entropy of a state $\rho\in\mathcal{S}(\hr)$ is given by
\begin{equation}S(\rho):=-\textrm{tr}(\rho \log\rho),\end{equation}
where $\log(\cdot)$ denotes the base two logarithm which is used throughout the paper (accordingly, $\exp(\cdot)$ is reserved for the base two exponential). For two states $\rho, \sigma \in \mathcal{S}(\hr)$, the quantum relative entropy is defined by
\begin{align}
 D(\rho||\sigma) := 
    \begin{cases}  \tr(\rho\log \rho - \rho\log\sigma) 	& \text{if}\ \ker \sigma \subseteq \ker \rho  \\
						+\infty & \text{else.} \end{cases} 
\end{align}
The Holevo information is for a given channel $W \in CQ(\mathbf{X},\hr)$ and input probability distribution $p \in \pr(\mathbf{X})$ defined by
\begin{align}
 \chi(p, W) := S(\overline{W}) - \sum_{x \in \mathbf{X}} p(x) S(W(x)) 
	     = \sum_{x \in \mathbf{X}} p(x) D(W(x)||\overline{W}), \label{def:holevo-chi}
\end{align}
where $\overline{W}$ is defined by $\overline{W} := \sum_{x \in \mathbf{X}} p(x) W(x)$. This quantity is concave w.r.t. the input probability distribution and convex w.r.t. the channel. Its concavity property follows directly from the concavity of the von Neumann entropy, its convexity in the channel is by joint convexity of the quantum relative entropy.
For an arbitrary set $\mathcal{W} \subset CQ(\mathbf X,\hr)$ we denote its convex hull by $\conv(\mathcal{W})$ (for the  definition of the convex hull, \cite{webster} is a useful reference). 
In fact, for a set $\mathcal{W}:= \{W_s\}_{s \in \bS}$
\begin{align}\label{eq:conv-hull}
 \conv(\mathcal{W})=\left\{W_{q}\in CQ(\mathbf X,\hr): W_q=\sum_{s\in \bS}q(s)W_s,\ q \in\mathfrak{P}(\bS),   |\supp(q)|< \infty \right\},
\end{align}
because of Carath\'eodory's Theorem.
\end{section}

\begin{section}{\label{sec:Definitions}Definitions}
\subsection{\label{subsec:compound-definitions}The compound classical-quantum channel}
Let $\W\subset CQ(\mathbf X,\hr)$. The memoryless compound cq-channel associated with $\W$ is given by the family 
$\{W^{\otimes l}\}_{l\in\nn,W\in \W}$. With slight abuse of notation it will be denoted $\W$ or, if necessary, 'the compound cq-channel $\W$' for short. In the remainder, 
using arbitrary index sets $T$, we will often write $\W=\{W_t\}_{t\in T}$ to enhance readability. Before we continue, let us put a brief remark in order to explain why this 
subsection contains no definition of random codes (while subsection \ref{subsec:avcqc-definitions} does):
\begin{remark}
We abstain from defining random codes for compound cq-channels, the reason for this being that they do offer no increase in capacity. For the reader interested in the topic, we 
briefly outline one way of arriving at this conclusion.\\
First, the capacity of compound channels, seen as a function from the power set of the set of channels with given input and output systems to the reals, is continuous (this can fact can be proven by an argument very similar to the one given for compound quantum channels in Sect. 8 of \cite{bbn-2} together with continuity of the single channel classical capacity, cf. \cite{leung-smith}). This 
allows for an arbitrarily good (speaking in terms of their capacity) approximation of infinite compound cq-channels by finite ones, so that we can restrict our discussion to finite 
compound cq-channels.\\
Second, given such a finite compound cq-channel $\{W_t\}_{t\in T}$ and a sequence of random codes which achieve a given rate $r$ with asymptotically vanishing average error, 
we may simply use it for the memoryless cq-channel $\overline W:=\frac{1}{|T|}\sum_{t\in T}W_t$. Since the average error is a convex function of the channel, this 
implies the existence of a sequence of deterministic codes at the same asymptotic rate with vanishing average error for $\overline W$.\\
Using affinity of the average error criterion once more, we see that the very same sequence of deterministic codes also has vanishing average error for the cq-compound 
channel $\{W_t\}_{t\in T}$, only with a slightly slower convergence. As in the definition of $\overline W$, the assumption that $|T|<\infty$ holds is crucial at this point of the argument. This shows that random codes 
cannot have higher asymptotic rates than deterministic ones, if one insists on asymptotically vanishing average error.\\
For the maximal error criterion, it is enough to note that both the random and the deterministic capacity for transmission of messages over a compound cq-channel using that criterion 
are upper bounded by the respective capacities for the average error criterion. 
\end{remark}
\begin{definition}
An \emph{$(l,M_l)$-code} for message transmission over a compound cq-channel $\mathcal{W} \subset CQ(\mathbf{X},\hr)$ is a family $(x^l_m,D_m^l)_{m=1}^{M_l}$, 
where $x^l_1,\ldots,x^l_{M_l}\in \mathbf X^l$ and $(D_1^l,\ldots,D_{M_l}^l)\in \mathcal{M}_{M_l}(\hr^{\otimes l})$.
\end{definition}
\begin{definition}
For $\lambda \in [0,1)$, a non-negative number $R$ is called a \emph{$\lambda$-achievable rate} for transmission of messages over the compound cq-channel $\W=\{W_t\}_{t\in T}$ using 
the average error criterion if there is a sequence $\{(u_m^l,D_m^l)_{m=1}^{M_l}\}_{l \in \nn}$ of $(l,M_l)$-codes with
\begin{align*}
 &\liminf_{l\to\infty}\frac{1}{l}\log M_l\ge R\qquad \textrm{and} \\
 &\underset{l\to\infty}{\limsup}\ \sup_{t\in T} \  \frac{1}{M_l}\sum_{m=1}^{M_l} \tr(W_t^{\otimes l}(u_m^l)(\eins_{\hr^{\otimes l}}-D_m^l)) \leq  \lambda.
\end{align*}
\end{definition}

\begin{definition}
For $\lambda \in [0,1)$, a non-negative number $R$ is called a \emph{$\lambda$-achievable rate} for transmission of messages over the compound cq-channel $\W=\{W_t\}_{t\in T}$ using the maximal error criterion if there is a sequence $\{(u_m^l,D_m^l)_{m=1}^{M_l}\}_{l \in \nn}$ of $(l,M_l)$-codes with
\begin{align*}
 &\liminf_{l\to\infty}\frac{1}{l}\log M_l\ge R\qquad \textrm{and} \\
 &\underset{l\to\infty}{\limsup}\ \sup_{t\in T} \ \max_{m \in [M_l]} \tr(W_t^{\otimes l}(u_m^l)(\eins_{\hr^{\otimes l}}-D_m^l)) \leq  \lambda.
\end{align*}
\end{definition}
\begin{definition}
For $\lambda \in [0,1)$, the \emph{$\lambda$-capacity} for message transmission using the average error criterion of a compound cq-channel $\mathcal{W}$ is given by
\begin{eqnarray}
\overline C_{C}(\mathcal{W},\lambda):=\sup\left\{R:\begin{array}{l}R\textrm{\ is\ a\ }\lambda\textrm{-achievable\ rate \ for}\\ \textrm{transmission\ of\ messages\ over\ }\mathcal{W}\\
 \textrm{using\ the\ average\ error\ probability\ criterion}\end{array}\right\}.
\end{eqnarray}
The number $\overline C_{C}(\mathcal{W},0)$ is called the \emph{weak capacity} for message transmission using the average error criterion of $\mathcal{W}$ and abbreviated $\overline C_{C}(\mathcal{W})$.
\end{definition}
\begin{definition}
For $\lambda \in [0,1)$, the \emph{$\lambda$-capacity} for message transmission using the maximal error criterion of a compound cq-channel $\mathcal{W}$ is given by
\begin{eqnarray}
C_{C}(\mathcal{W},\lambda):=\sup\left\{R:\begin{array}{l}R\textrm{\ is\ a\ }\lambda\textrm{-achievable\ rate\ for\ transmission} \\ \textrm{of\ messages\ over\ }\mathcal{W} \\
 \textrm{using\ the\ maximal\ error\ probability\ criterion}\end{array}\right\}.
\end{eqnarray}
The number $C_{C}(\mathcal{W},0)$ is called the \emph{weak capacity} for message transmission using the maximal error criterion of $\mathcal{W}$ and abbreviated $C_{C}(\mathcal{W})$.
\end{definition}
\subsection{\label{subsec:avcqc-definitions}The arbitrarily varying classical-quantum channel}
Let $\A\subset CQ(\mathbf X,\hr)$. In the remainder we will write $\A=\{A_s\}_{s\in \bS}$, where $\bS$ denotes an index set, in order to enhance readability. We also set 
\begin{align}
A_{s^l}:=\otimes_{i=1}^lA_{s_i}.
\end{align}
The arbitrarily varying classical-quantum channel associated with $\A$ is given by the family 
$\{A_{s^l}\}_{s^l\in \bS^l,l\in\nn}$. Again, with slight abuse of notation it will be denoted $\A$ or, if necessary, 'the AVcqC $\A$' for short.\\
In this work, we will always consider the set $\mathbf S$ to be finite. Generalizations of our results to the case of arbitrary sets can be done by standard techniques (see \cite{abbn}). 
We will now define random codes and the random capacity emerging from them. In order to do so, we have to clarify a few things. 
\\A code for an AVcqC
$\A$ will, for some choice of $l,N\in\nn$, 
be given by a probability measure $\mu_l$ on the set $((\mathbf X^l)^{N}\times\M_N(\hr^{\otimes l}),\Sigma_l)$, where $\Sigma_l$ is a suitably chosen
sigma-algebra.\\
It has to be taken care that a function $f$ defined by
$((x_1^l,\ldots,x_N^l),(D_1^l,\ldots,D_N^l))\mapsto\min_{s^l\in\bS^l}\frac{1}{N}\sum_{i=1}^N\tr\{W_{s^l}(x^l_i)D_i^l\}$ is measurable w.r.t. 
$\Sigma_l$. Also, in order to define deterministic codes later, $\Sigma_l$ has to contain all the singleton sets. In the remainder, we shall assume that
such a choice is always 
made.\\
An explicit example of such a sigma-algebra is given by the Borel sigma-algebra 
defined using the topology induced by the metric $((x,D),(x',D'))\mapsto(1-\delta(x,x'))+\|D-D'\|_2$ where $\delta(x,x)=1\ \forall x\in\mathbf X$ and
equal to zero else, and 
for sake of simplicity, we set $l=N=1$. Finally, we note that the function $f$ mentioned above is continuous w.r.t. to that metric.\\
In the following definitions, let $\lambda\in[0,1)$.
\begin{definition}\label{def:avc_rand-code} 
An $(l,M_l)$-\emph{random code} for message transmission over $\A=\{A_s\}_{s\in \bS}$ is a probability measure $\mu_l$ on $(\mathbf
(X^l)^{M_l}\times\M_N(\hr^{\otimes l}),\Sigma_l)$. In order to shorten our notation, we write elements of $(X^l)^{M_l}\times\M_N(\hr^{\otimes l})$ in the form 
$(x^l_i,D^l_i)_{i=1}^{M_l}$.
\end{definition}
\begin{definition}\label{def:avc_det-code} 
An \emph{$(l,M_l)$-deterministic code} for message transmission over $\A=\{A_s\}_{s\in \bS}$ is given by a random code for message transmission over
$\A$ with $\mu_l$ assigning probability one to a singleton set.
\end{definition}
\begin{definition}\label{def:avc_rand-achievability}
A non-negative number $R$ is called \emph{$\lambda$-achievable} for transmission of messages over the AVcqC $\A=\{A_s\}_{s\in \bS}$ with random codes using the average
error criterion if 
there is a sequence $(\mu_l)_{l\in\nn}$ of $(l,M_l)$-random codes such that the following two lines are true:
\begin{equation}
\liminf_{l\to\infty}\frac{1}{l}\log M_l\ge R
\end{equation}
\begin{equation}
\limsup_{l\to\infty}\max_{s^l\in\bS^l} \int \frac{1}{M_l} \sum_{i=1}^{M_l} \tr\left(A_{s^l}(x_i^l)(\eins_{\hr^{\otimes l}} - D_i^l)\right)\ d\mu_l((u_i^l,D_i^l)_{i=1}^{M_l}) \leq \lambda.
\end{equation}
\end{definition}
\begin{definition}\label{def:avc_det-achievability}
A non-negative number $R$ is called \emph{$\lambda$-achievable} for transmission of messages over the AVcqC $\A=\{A_s\}_{s\in \bS}$ with deterministic codes using the
average error criterion if 
it is $\lambda$-achievable with random codes by a sequence $(\mu_l)_{l\in\nn}$ which are deterministic codes.
\end{definition}
\begin{definition}
The \emph{$\lambda$-capacity} for message transmission using random codes and the average error criterion of an AVcqC $\A$ is given by
\begin{align}
\overline C_{\textup{A,r}}(\A,\lambda):=\sup\left\{R:\begin{array}{l}R\textrm{\ is\ a\ }\lambda\textrm{-achievable\ rate\ for\ transmission\ of} \\ \textrm{messages\ over\ }\A \textrm{\ with\ random\ codes} \\ \textrm{using\ the\ average \ error\ probability\ criterion}\end{array}\right\}.
\end{align}
The number $\overline C_{\textup{A,r}}(\A,0)$ is called the \emph{weak capacity} for message transmission using random codes and the average error criterion of $\A$ and abbreviated $\overline C_{\textup{A,r}}(\A)$.
\end{definition}
\begin{definition}
The \emph{$\lambda$-capacity} for message transmission using deterministic codes and the average error criterion of an AVcqC $\A$ is given by
\begin{align}
\overline C_{\textup{A,d}}(\A,\lambda):=\sup\left\{R:\begin{array}{l}R\textrm{\ is\ a\ }\lambda\textrm{-achievable\ rate\ for\ transmission\ of}\\ \textrm{messages\ over\ }\A\textrm{\
with \ deterministic\ codes} \\ \textrm{using\ the\ average\ error\ probability\ criterion}\end{array}\right\}.
\end{align}
The number $\overline C_{\textup{A,d}}(\A,0)$ is called the \emph{weak capacity} for message transmission using deterministic codes and the average error criterion of $\A$ and abbreviated $\overline C_{\textup{A,d}}(\A)$.
\end{definition}
\begin{definition}\label{def:avcqc-achievability-maximal-error}
A non-negative number $R$ is called \emph{$\lambda$-achievable} for transmission of messages over the AVcqC $\A=\{A_s\}_{s\in \bS}$ with deterministic codes using the
maximal error probability 
criterion if there is a sequence of $(l,M_l)$-random codes with each $\mu_l$ being a deterministic code such that the following two lines are true:
\begin{equation}
\liminf_{l\to\infty}\frac{1}{l}\log M_l\ge R
\end{equation}
\begin{equation}
\limsup_{l\to\infty}\max_{s^l\in\bS^l}\max_{i=1,\ldots,M_l}\int \tr\left(A_{s^l}(x_i^l)(\eins_{\hr^{\otimes l}} - D_i^l)\right)\ d\mu_l((u_i^l,D_i^l)_{i=1}^{M_l})\leq\lambda.
\end{equation}
\end{definition}
\begin{definition}
The \emph{$\lambda$-capacity} for message transmission using deterministic codes and the maximal error probability criterion of an AVcqC $\A$ is given by
\begin{align}
C_{\textup{A,d}}(\A,\lambda):=\sup\left\{R:\begin{array}{l}R\textrm{\ is\ a\ }\lambda\textrm{-achievable\ rate\ for\ transmission}\\ \textrm{of\ messages\ over\ }\A \textrm{\ with deterministic\ codes} \\ \textrm{using\ the\ maximal\ error\ probability\ criterion}\end{array}\right\}.
\end{align}
The number $C_{\textup{A,d}}(\A,0)$ is called the  \emph{weak capacity} for message transmission using deterministic codes and the maximal error criterion of $\A$ and abbreviated $C_{\textup{A,d}}(\A)$.
\end{definition}
The following definition will turn out to be useful to decide whether a given $AVcqC$ has nonzero capacity for transmission of messages using average error criterion and 
deterministic codes.
\begin{definition}{\label{def:symmetrizability-for-maximal-error}}
 Let $\A=\{A_s\}_{s\in\bS}\subset CQ(\mathbf X,\hr)$ be an AVcqC. If, for every $x,x'\in\mathbf X$, we have
\begin{equation}\label{eq:equiv-max-1}
\conv(\{\A_s(x)\}_{s\in\bS})\cap \conv(\{\A_s(x')\}_{s\in\bS})\neq\emptyset,
\end{equation}
then $\A$ is called \emph{m-symmetrizable}.
\end{definition}
\subsection{Zero-error capacity\label{subsec:def:zero-error}}
\begin{definition}
An \emph{$(l,M_l)$ zero-error code} for a stationary memoryless cq-channel defined by $V\in CQ(\mathbf X,\hr)$ is given by a family
$(x^l_i,D_i^l)_{i=1}^{M_l}$, where 
$x^l_1,\ldots,x^l_{M_l}\in \mathbf X^l$ and $(D_1^l,\ldots,D_{M_l}^l)\in \mathcal{M}_{M_l}(\hr^{\otimes l})$ 
satisfy $\tr(V^{\otimes l}(x^l_i)D_i^l)=1$ for every $i\in[M_l]$.
\end{definition}
\begin{definition}
The \emph{zero-error capacity} for message transmission over the cq-channel $V\in CQ(\mathbf X,\hr)$ is given by
\begin{equation}
 C_0(V):=\lim_{l\to\infty}\frac{1}{l}\log\max\{M_l:\exists\ (l,M_l)\ \mathrm{zero-error\ code\ for\ }V\}.
\end{equation}
\end{definition}
\end{section}

\begin{section}{Main Results\label{section:main-results}}
We now enlist the main results contained in this work. We will not state the results obtained in Subsection \ref{subsec:zero-error}. These evolve around the 
relation between zero-error capacities and arbitrarily varying channels. They include both message transmission and entanglement transmission. Rather than stating a positive 
result, in this section we argue that certain straightforward quantum analogues of results that are valid in the classical theory do not hold. As always, this is a delicate 
task that involves much more than just embedding a commutative subalgebra into a non-commutative one. We therefore encourage the reader to consider this last subsection as something 
that should be read separately and in one piece.\\
Our first result is the following.
\begin{theorem}[cq Compound Coding Theorem]\label{comp_cq_direct_part}
For every compound cq-channel $\mathcal{W} \in CQ(\mathbf{X},\hr)$ it holds
\begin{align}
 \overline{C}_{C}(\mathcal{W}) = \underset{p \in \pr(\mathbf{X})}{\max}\; \underset{W \in \mathcal{W}}{\inf} \chi(p,W).
  \label{comp_cq_capacity}
\end{align}
\end{theorem}
In subsection \ref{subsec:The Ahlswede-Dichotomy for AVcqCs}, an analogue of the Ahlswede dichotomy from \cite{ahlswede-elimination} for arbitrarily varying classical-quantum channels will be derived.
This statement has originally been obtained by Ahlswede and Blinovsky in \cite{ahlswede-blinovsky}. The precise mathematical formulation reads as follows.
\begin{theorem}[Ahlswede-Dichotomy for AVcqCs]\label{ahlswede-dichotomy-for-avcqcs}
 Let $\A=\{A_s\}_{s\in\bS}\subset CQ(\bX,\hr)$ be an AVcqC. Then
\begin{align}
1)&\qquad \overline C_{A,r}(\A)=\overline C_C(\conv(\A)) \label{ahlswede-dichotomy-for-avcqcs_1}\\
2)&\qquad\textup{If\ }\overline C_{A,d}(\A)>0\textup{,\ then\ }\overline C_{A,d}(\A)=\overline C_{A,r}(\A). \label{ahlswede-dichotomy-for-avcqcs_2}
\end{align}
\end{theorem}
Also, this section contains the following statement, which asserts, that every sequence of random codes whith error strictly smaller than 1 for all but finitely many blocklenghts will not 
achieve rates higher than the rightmost term in (\ref{ahlswede-dichotomy-for-avcqcs_1}).
\begin{theorem}[Strong converse]\label{avcqc_strong_converse}
Let $\mathcal{A}:= \{A_s\}_{s \in \mathbf{S}}$ be an AVcqC. For every $\lambda \in [0,1)$ 
\begin{align}
  \overline{C}_{A,r}(\mathcal{A}, \lambda) \leq \overline{C}_C(\conv(\mathcal{A}))
\end{align}
holds.
\end{theorem}
\begin{remark}
The result can be gained for arbitrary (infinite) AVcqCs with only trivial modifications of the proof given below.
\end{remark}
In the next subsection \ref{subsec:maximal-and-zero-error}, we show that the capacity for message transmission over an AVcqC using
deterministic codes and the maximal error probability criterion is zero if and only if the AVcqC is $m-symmetrizable$.\\
This is an analog of \cite[Theorem 1]{kiefer-wolfowitz}. It can be formulated as follows.
\begin{theorem}\label{theorem:c-det=0-for-maximal-error}
Let $\A=\{A_s\}_{s\in\bS}\subset CQ(\mathbf X,\hr)$ be an AVcqC. Then $C_{A,d}(\A)$ is equal to zero if and only if $\A$ is m-symmetrizable.
\end{theorem}
\end{section}
\begin{section}{\label{sec:compound}Compound cq-channels}
In this section, we consider compound cq-channels and give a rigourous proof for the achievability part of the coding theorem under 
the average error 
criterion together with a weak converse. The channel coding problem for compound cq-channels was treated, restricted to achievability, 
by Datta and Hsieh \cite{datta-hsieh} 
for a certain class of compound channels, and Hayashi \cite{hayashi08}. In our proof, we exploit the close relationship between channel
 coding and hypothesis testing 
which was utilized by Hayashi and Nagaoka \cite{hayashi03} before. With focus set on the maximal error criterion, the compound cq channel
 coding theorem was proven in 
\cite{bb-compound} already where also a strong converse theorem was proven for this setting.\\
For orientation of the reader we sketch the contents of this section. In Lemma \ref{lemma:statistical-test-to-code} we reduce the problem of finding good channel codes 
for a finite compound channel to the problem of finding good hypothesis tests for certain quantum states generated by this channel. The existence of hypothesis tests with 
a performance sufficient for our purposes is shown in Lemma \ref{lemma:existence-of-statistical-test}. In order to establish the coding theorem for arbitrary compound 
channels, we recall some approximation results in Lemma \ref{lemma:approximation}. With these preparations, we are able to prove the direct part of the coding theorem. 
Additionally, we give a proof of the weak converse (for which we utilize the strong converse result for the maximal error criterion given in \cite{bb-compound} in 
Theorem \ref{comp_cq_direct_part}). A strong converse for coding under the average error criterion does not hold in general for compound cq-channels (for further information, 
see Remark \ref{remark10}).\\
We consider a compound channel $\mathcal{W} := \{W_t\}_{t\in T}\subset CQ(\mathbf{X},\hr)$ where T is a finite index
set. We fix an orthonormal basis $\{e_x\}_{x \in \mathbf{X}}$ in $\cc^{|\mathbf{X}|}$. For $\mathcal{W}$ and a given input probability distribution $p \in \pr(\mathbf{X})$ we define for every $t \in T$ states 
\begin{align}
 \rho_t := \sum_{x \in \mathbf{X}} p(x) \ket{e_x}\bra{e_x} \otimes W_t(x), \hspace{0.7cm} \text{and} \hspace{1.5cm}
 \hat{\sigma}_t := p \otimes \sigma_t, \label{comp_cq_rho_t}
\end{align}
on $\cc^{|\mathbf{X}|} \otimes \hr$, where $p$ and $\sigma_t$ are defined by
\begin{align}
	      p := \sum_{x \in \mathbf{X}} p(x) \ket{e_x} \bra{e_x}, \hspace{0.7cm} \text{and} \hspace{1.5cm}
       \sigma_t := \sum_{x\in \mathbf{X}} p(x) W_t(x) \label{comp_cq_sigma_t}.
\end{align}
With some abuse of notation, we use the letter $p$ for the probability distribution as well as for the according quantum state defined above. Moreover, we define for every $l \in \nn$ states
\begin{align}
\rho_l &:= \frac{1}{|T|} \sum_{t \in T} v_l\rho_t^{\otimes l}v_l^\ast \label{comp_cq_rho_l} \\
\tau_l &:= \frac{1}{|T|} \sum_{t \in T} v_l\hat{\sigma}_t^{\otimes l}v_l^\ast 
	 = p^{\otimes l} \otimes \frac{1}{|T|}\sum_{t \in T}\sigma_t^{\otimes l} \label{comp_cq_tau_l} 
\end{align}
where $v_l: (\cc^{|\mathbf{X}|} \otimes \hr)^{\otimes l} \rightarrow (\cc^{|\mathbf{X}|})^{\otimes l} 
\otimes \hr^{\otimes l}$ is the ismorphism permuting the tensor factors. The next lemma is a variant of a result
by Hayashi and Nagaoka in \cite{hayashi03}, which states that good hypothesis tests imply good message transmission codes for the average error criterion. Here it is formulated and proven for the states $\rho_l$ and $\tau_l$.
\begin{lemma}\label{lemma:statistical-test-to-code}
Let $\mathcal{W} := \{W_t\}_{t \in T} \subset CQ(\mathbf{X}, \hr)$ be a compound cq-channel with $|T| < \infty$, 
$p \in \pr(\mathbf{X})$, and $l \in \nn$. Let further $\rho_l$, $\tau_l$ be the states associated to $\mathcal{W}$,$p$
as defined in (\ref{comp_cq_rho_l}) and (\ref{comp_cq_tau_l}). If for $\lambda \in
[0,1]$, and $a > 0$ exists a projection $q_l \in 
\bo((\cc^{|\mathbf{X}|})^{\otimes l} \otimes \hr^{\otimes l})$ which fulfills the conditions
\begin{enumerate}
 \item $\tr(q_l \rho_l) \geq 1 - \lambda$ 
 \item $\tr(q_l \tau_l) \leq 2^{-la}$,
\end{enumerate}
then for any $\gamma$ with $a \geq \gamma > 0$ and $M_l := \lfloor 2^{l(a -\gamma)} \rfloor$ there is an $(l,M_l)$-code $(x_m^l, D_m^l)_{m \in [M_l]}$ 
with 
\begin{align}
 \max_{t \in T}\ \frac{1}{M_l}\sum_{m=1}^{M_l}\tr(W_t^{\otimes l}(x_m^l)(\eins_{\hr^{\otimes l}}-D_m^l)) \leq |T|(2 \lambda + 4\cdot 2^{-l \gamma})
\end{align}
\end{lemma}
The following operator inequality is a crucial ingredient in the proof of the lemma above, it was given in a more general form by Hayashi and Nagaoka in \cite{hayashi03}. 
\begin{lemma}\label{lemma:hayashi-nagaoka}
Let $a,b \in \bo(\hr)$ be operators on $\hr$ with $0 \leq a \leq 1$ and $b \geq 0$. Then
\begin{align}
\eins_\hr - (a+b)^{-\frac{1}{2}} a (a+b)^{-\frac{1}{2}} \leq 2(\eins_\hr - a) + 4b,
\end{align}
where $(\cdot)^{-1}$ denotes the generalized inverse.
\end{lemma}
\begin{proof}
See Lemma 2 in \cite{hayashi03}.\qed
\end{proof}
\begin{proof}[of Lemma \ref{lemma:statistical-test-to-code}] 
Let $l \in \nn$, $q_l$ a projection such that the assumptions of 
the lemma are fulfilled, and $ \gamma$ a number with $0<\gamma\leq a$. According to the assumptions, $q_l$ takes the form
\begin{align}
 q_l = \sum_{x^l \in \mathbf{X}^l} \ket{e_{x^l}}\bra{e_{x^l}} \otimes q_{x^l},
\end{align}
where $q_{x^l} \in \bo(\hr^{\otimes l})$ is a projection for every $x^l \in \mathbf{X}^l$. Set 
$M_l := \lfloor 2^{l(a-\gamma)} \rfloor$, and let $U_1,...,U_{M_l}$ be i.i.d. random variables with values in 
$\mathbf{X}^l$, each distributed according to the $l$-fold product $p^{\otimes l}$ of the given distribution $p$. 
We define a random operator
\begin{align}
 D_m := \left(\sum_{n=1}^{M_l} q_{U_n}\right)^{-\frac{1}{2}} q_{U_m} 
	\left(\sum_{n=1}^{M_l} q_{U_n} \right)^{-\frac{1}{2}} \label{comp_cq_dec_op}
\end{align}
for every $m \in [M_l]$ (we omit the superscript $l$ here), where again generalized inverses are taken. The particular form of 
the decoding operators $D_1,...,D_{M_l}$ in eq. (\ref{comp_cq_dec_op}) guarantees, that 
\begin{align*}
 \sum_{m=1}^{M_l} D_m \leq \eins_{\hr^{\otimes l}}
\end{align*}
holds for every outcome of $U_1,...,U_{M_l}$, and therefore $(U_m, D_m)_{m \in [M_l]}$ is a random code of size $M_l$.
The remaining task is to bound the expectation value of the average error of this random code. We introduce an 
abbreviation for the average of the channels in $\mathcal{W}$ by
\begin{align*}
 \overline{W}^{l}(\cdot) := \frac{1}{T}\sum_{t=1}^T W_t^{\otimes l}(\cdot).
\end{align*}
The error probability of the random code is bounded as follows. By virtue of Lemma \ref{lemma:hayashi-nagaoka}, 
\begin{alignat}{3}
 \erw\left[\tr\left(\overline{W}^{l}(U_m)(\eins_{\hr^{\otimes l}}-D_m)\right)\right] &\, \leq \, 
	    && 2 \ \erw\left[\tr\left(\overline{W}^{l}(U_m)(\eins_{\hr^{\otimes l}} - q_{U_m})\right)\right] \nonumber \\
					  &  && + 4 \cdot \sum_{\substack{m \in [M_l]: \\ n \neq m}}
						   \erw\left[\tr\left(\overline{W}^{l}(U_m)q_{U_n}\right)\right]& 
    \label{comp_cq_erw_rn_1}
\end{alignat}
holds. The calculation of the expectation values on the r.h.s. of the above equation is straightforward, we obtain 
for every $m\in [M_l]$
\begin{align}
 \erw[\tr(\overline{W}^{l}(U_m)(\eins_{\hr^{\otimes l}} - q_{U_m}))] 
	&= \tr(\rho_l(\eins_{\hr^{\otimes l}}- q_l)), \label{comp_cq_erw_rn_1_1}
\end{align}
and, for $n \neq m$,
\begin{align}
 \erw\left[\tr\left(\overline{W}^{l}(U_m)q_{U_n}\right)\right] = \tr(\tau_l q_l) \label{comp_cq_erw_rn_1_2}.
\end{align}
Together with the assumptions of the lemma, eqns. (\ref{comp_cq_erw_rn_1_1}) and (\ref{comp_cq_erw_rn_1_2}) imply
\begin{align*}
\erw\left[\tr\left(\frac{1}{|T|}\sum_{t \in T} W^{\otimes l}_t(U_m)(\eins_{\hr^{\otimes l}} -D_m)\right)\right] 
    &\leq 2 \lambda + 4 \cdot M_l \cdot 2^{-la}\\
    &\leq 2\lambda + 4 \cdot 2^{-l\gamma}
\end{align*}
Because this error measure is an affine function of the channel we conclude, that there exists a cq-code 
$(x_m^l, D_m)_{m=1}^{M_l}$ for $\mathcal{W}$ with average error bounded by
\begin{align}
 \frac{1}{M_l}\sum_{m=1}^{M_l}\tr(W_t^{\otimes l}(x_m^l)(\eins_{\hr^{\otimes l}}-D_m)) \leq |T|(2 \lambda + 4 \cdot 2^{-l\gamma})
\end{align}
for every $t \in T$, which is what we aimed to prove.\qed
\end{proof}
The next two lemmata contain facts which are important for later considerations. The first lemma presents a bound on 
the cardinality of the spectrum of operators on a tensor product space which are invariant under permutations of the 
tensor factors. The group $S_l$ of permutations on $[l]$ is, on $\hr^{\otimes l}$, represented by defining (with slight abuse of notation) 
for each $\sigma \in S_l$ the unitary operator $\sigma \in \bo(\hr^{\otimes l})$ 
\begin{align}
 \sigma(v_1 \otimes ... \otimes v_l) := v_{\sigma^{-1}(1)} \otimes ... \otimes v_{\sigma^{-1}(l)}. \label{tensor_prod_rep}
\end{align}
for all product vectors $v_1 \otimes ... \otimes v_l \in \cc^l$ and linear extension to the whole space $\cc^{\otimes l}$. 
\begin{lemma}\label{fact-1} Let $Y\in\mathcal B(\hr^{\otimes l})$ ($d:=\dim\hr\geq2$) satisfy $\sigma Y=Y\sigma$ for every permutation $\sigma\in S_l$. Then
\begin{align}|\spec(Y)|\leq(l+1)^{d^2}.\end{align}
\end{lemma}
\begin{proof}
It is clear that, under the action of $S_l$, $\hr^{\otimes l}$ decomposes
into a finite direct sum $\hr^{\otimes
l}=\oplus_{i=1}^M\oplus_{j=1}^{m_i}\hr_{i,j}$, where the $\hr_{i,j}$ are
irreducible subspaces of $S_l$, $m_i\in\nn$ their multiplicity and
$M\in\nn$. Moreover, $\hr_{i,j}\simeq\hr_{i,k}$ f.a. $i\in [M]$,
$j,k\in[m_i]$ and to every such choice of
indices there exists a linear operator
$Q_{i,j,k}:\hr_{i,k}\mapsto\hr_{i,j}$ such that $\sigma
Q_{i,j,k}=Q_{i,j,k}\sigma$ f.a. $\sigma\in S_l$.\\
Let us write $Y=\sum_{i,j}Y_{i,m,j,n}$, where
$Y_{i,m,j,n}:\hr_{j,n}\mapsto\hr_{i,m}$. Then according to Schur's lemma,
$Y_{i,m,j,n}=0$, ($i\neq j$) and $Y_{i,m,i,n}=c_{i,m,n}Q_{i,m,n}$ for all
valid
choices of indices and unique complex numbers $c_{i,m,n}\in\mathbb C$.\\
Thus, defining the self-adjoint operators
$Y_i:=\sum_{m,n=1}^{m_i}c_{i,m,n}Q_{i,m,n}$, we see that
\begin{align}
Y=\sum_{i=1}^MY_i
\end{align}
holds. Obviously, $Y_{i,m,i,m}=\eins_{\hr_{i,m}}$. Thus, with an
appropriate choice of bases in every single one of the $\hr_{i,m}$ and
defining the matrices $C_i$ by $(C_i)_{mn}:=c_{i,m,n}$,
we can write a matrix representation $\tilde Y_i$ of $Y_i$ as $\tilde
Y_i=C_i\otimes\eins_{\mathbb C^{\dim(\hr_{i,1})}}$.\\
Clearly then, each of the $Y_i$ can have no more than $m_i$ different
eigenvalues. Since $\supp(Y_i)\perp\supp(Y_j)$ ($i\neq j$), we get
\begin{align}
|\spec(Y)|\leq\sum_{i=1}^Mm_i.
\end{align}
Now, taking a look at \cite{christandl}, equation (1.22), we see that
$m_i\leq(l+1)^{d^2/2}$ holds. The number $M$ is the number of different
Young tableaux occuring in the representation
of $S_l$ on $\hr^{\otimes l}$ and obeys the bound $M\leq N_T([d]^l)$,
where $N_T([d]^l)$ is the number of different types on $[d]^l$, that
itself obeys
$N_T([d]^l)\leq(l+1)^d$ (Lemma 2.2 in \cite{csiszar-koerner}). For
$d\geq2$ we thus have
\begin{align}
|\spec(Y)|\leq\sum_{i=1}^Mm_i\leq(l+1)^{d^2/2}(l+1)^d\leq(l+1)^{d^2}.
\end{align}\qed
\end{proof}
Lemma \ref{lemma:existence-of-statistical-test} provides the result which will, together with Lemma \ref{lemma:statistical-test-to-code}, imply the existence of optimal codes for $\mathcal{W}$. We give a proof which is based on an idea of Ogawa and Hayashi which originally appeared in \cite{ogawa01}. An important ingredience of their proof is the operator inequality stated in the following lemma.
\begin{lemma}[\cite{hayashi01}]\label{fact-2}
Let $\chi$ be a state on on a Hilbert space $\kr$, and $\mathcal{M}:=\{P_k\}_{k=1}^K \subset \bo(\kr)$ be a collection of projections on $\kr$ with $\sum_{k=1}^K P_k= \eins_{\kr}$.
Then the operator inequality 
\begin{align}
 \chi \leq K \cdot \sum_{k=1}^K P_k \chi P_k
\end{align}
holds.
\end{lemma}
\begin{lemma}\label{lemma:existence-of-statistical-test}
For every $\delta > 0$, finite compound cq-channel $\mathcal{W} := \{W_t\}_{t \in T} \subset CQ(\mathbf{X},\hr)$ and $p \in \pr(\mathbf{X})$ there exists a constant $\tilde{c}$, such that for every sufficiently large $l \in \nn$ there exists a projection $q_{l,\delta} \in \bo((\cc^{|\mathbf{X}|})^{\otimes l}\otimes \hr^{\otimes l})$ which fulfills
\begin{enumerate}
 \item $\tr(q_{l,\delta}\rho_l) \geq 1 - |T| \cdot 2^{-l\tilde{c}}$, \text{and}
 \item $\tr(q_{l,\delta}\tau_l) \leq 2^{-l(a - \delta)}$
\end{enumerate}
where $\rho_l, \tau_l$ are the states belonging to $\mathcal{W},p$ according to (\ref{comp_cq_rho_l}) and (\ref{comp_cq_tau_l}), and $a$ is defined by $a := \min_{t \in [T]} D(\rho_t|| p \otimes \sigma_t)$.
\end{lemma}
\begin{proof}
Let $\delta > 0$ be fixed, for $l \in \nn$, we have $\ran(\rho_l) \subseteq \ran(\tau_l) := \hr_l$, which allows 
us to restrict ourselves to $\hr_l$, where $\tau_l$ is invertible. For every $\varepsilon \in (0,1)$, we define a regularized version $\rho_{l, \varepsilon}$ to $\rho_l$ by 
\begin{align}
 \rho_{l,\varepsilon} := (1-\varepsilon) \rho_{l} + \varepsilon \tau_l. \label{rho_l_regularization}
\end{align}
These operators are invertible on $\hr_l$ and approximate $\rho_l$, i.e.
\begin{align}
 \|\rho_{l,\varepsilon} - \rho_{l} \|_1 \leq 2 \varepsilon. \label{fannes_abstand}
\end{align}
holds for every $\epsilon > 0$. We also define an operator 
\begin{align}
 \overline{\rho}_{l,\varepsilon} := \sum_{\lambda \in \spec(\tau_l) \setminus \{0\}} E_\lambda \rho_{l,\varepsilon} E_\lambda,
\end{align}
which is the pinching of $\rho_{l,\varepsilon}$ to the eigenspaces of $\tau_l$ (here $E_\lambda$ is the projection which 
projects onto the eigenspace belonging to the eigenvalue $\lambda$ for every $\lambda \in \spec(\tau_l)$). This definition
guarantees 
\begin{align}
 \tau_l\overline{\rho}_{l,\varepsilon} = \overline{\rho}_{l,\varepsilon}\tau_l. \label{comp_cq_hyp_commute}
\end{align}
With $a$ as assumed in the lemma, we define the operator
\begin{align}
 T_\varepsilon:= \overline{\rho}_{l,\varepsilon} - 2^{l(a -\delta)} \tau_l \label{te_epsilon}
\end{align}
with spectral decomposition 
\begin{align}
 T_{\varepsilon} = \sum_{\mu \in \spec(T_{\varepsilon})} \mu P_\mu.
\end{align}
The projection $q_{l,\delta}$ onto the nonnegative part of $T_{\varepsilon}$, defined by
\begin{align}
 q_{l,\delta} := \sum_{\mu \in \spec(T_{\varepsilon}): \mu \geq 0} P_\mu. \label{comp_cq_hyp_proj}
\end{align}
will now be shown to suffice the bounds stated in the lemma. Clearly, $q_{l,\delta}T_{\varepsilon}q_{l,\delta}$ is a positive semidefinite operator, therefore, with (\ref{te_epsilon}) 
the inequality
\begin{align}
 q_{l,\delta} \tau_l q_{l,\delta} \leq 2^{-l(a - \delta)} q_{l,\delta} \overline{\rho}_{l,\varepsilon} q_{l,\delta}. 
\label{comp_cq_hyp_proj_ineq}
\end{align}
is valid. Taking traces in (\ref{comp_cq_hyp_proj_ineq}) yields
\begin{align}
 \tr(q_{l,\delta}\tau_l) &\leq 2^{-l(a - \delta)} \tr(q_{l,\delta}\overline{\rho}_{l,\varepsilon}) \\
			 &\leq 2^{-l(a - \delta)}
\end{align}
which shows, that $q_{l,\delta}$ fulfills the second bound in the lemma. 
We shall now prove, that $q_{l,\delta}$ for $l$ large enough actually also suffices the first one. To this end we derive an upper bound 
on $\tr((\eins - q_{l,\delta}) \rho_{l,\varepsilon})$ for any given $\varepsilon>0$, which implies (together with 
(\ref{fannes_abstand})) a bound on $\tr((\eins - q_{l,\delta})\rho_{l})$. In fact it is sufficient to find an upper 
bound on $\tr((\eins - q_{l,\delta})\overline{\rho}_{l,\varepsilon})$, which can be seen as follows. Because 
$\overline{\rho}_{l,\varepsilon}$ and $\tau_l$ commute by construction (see eq. (\ref{comp_cq_hyp_commute})), 
$T_{\varepsilon}$ and $\tau_l$ commute as well. This in turn implies that $q_{l,\delta}$ commutes with the operators 
$E_1,...,E_{|\spec(\tau_l)|}$ in the spectral decomposition of $\tau_l$ which eventually ensures us, that
\begin{align}
 \tr((\eins_{\hr^{\otimes l}} - q_{l,\delta})\overline{\rho}_{l,\varepsilon}) = \tr((\eins_{\hr^{\otimes l}} - q_{l,\delta})\rho_{l,\varepsilon}) 
	\label{comp_cq_hyp_wbwob}
\end{align}
holds. For an arbitrary but fixed number $s \in [0,1]$ we have 
\begin{align}
 \tr((\eins_{\hr^{\otimes l}} - q_{l,\delta})\overline{\rho}_{l,\varepsilon}) 
	  & = \tr(\overline{\rho}_{l,\varepsilon}^{(1-s)}\overline{\rho}_{l,\varepsilon}^s(\eins_{\hr^{\otimes l}} - q_{l,\delta})) \\
	  & \leq 2^{-ls(a-\delta)}\tr(\overline{\rho}_{l,\varepsilon}^{(1-s)}\tau_l^s(\eins_{\hr^{\otimes l}} - q_{l,\delta})) 
	    \label{tr_mono} \\
	  & \leq 2^{-ls(a -\delta)} \tr(\overline{\rho}_{l,\varepsilon}^{(1-s)}\tau_l^s). \label{comp_cq_hyp_ineq_1}
\end{align}
The inequality in (\ref{tr_mono}) is justified by the following argument. Since $\overline{\rho}_{\varepsilon,l}$ and 
$\tau_l$ commute, they are both diagonal in the same orthonormal basis $\{g_i\}_{i=1}^d$, i.e. they have spectral 
decompositions of the form
\begin{align}
 \overline{\rho}_{l,\varepsilon} = \sum_{i=1}^d \chi_i \ket{g_i}\bra{g_i},\qquad \text{and}
      \hspace{0.3cm} \tau_l = \sum_{i=1}^d \theta_i \ket{g_i}\bra{g_i}. \label{comp_cq_hyp_spd}
\end{align}
Because $q_{l,\delta}$ projects onto the eigenspaces corresponding to nonnegative eigenvalues of $T_{\varepsilon}$, we have
\begin{align}
 \eins_{\hr^{\otimes l}} - q_{l,\delta} = \sum_{i \in N} \ket{g_i}\bra{g_i},
\end{align}
where the set $N$ is defined by $N := \{i \in [d]: \chi_i - 2^{l(a-\delta)}\theta_i < 0\}$.  It follows
\begin{align}
\chi_i^s \leq 2^{ls(a-\delta)} \theta_i^s
\end{align}
for all $i \in N$ and $s \in [0,1]$. This in turn implies, via (\ref{comp_cq_hyp_commute}) and (\ref{comp_cq_hyp_spd}),
\begin{align}
 \overline{\rho}_{l,\varepsilon}^s(\eins_{\hr^{\otimes l}} - q_{l,\delta}) \leq 2^{ls(a-\delta)} \tau_{l}^s (\eins_{\hr^{\otimes l}}-q_{l,\delta}),
\end{align}
which shows (\ref{tr_mono}). Combining eqns. (\ref{comp_cq_hyp_wbwob}) and (\ref{comp_cq_hyp_ineq_1}) we obtain
\begin{align}
 \tr((\eins_{\hr^{\otimes l}} - q_{l,\delta})\rho_{l,\varepsilon})	
    &\leq 2^{ls(a -\delta)} \tr(\overline{\rho}_{l,\varepsilon}^{(1-s)}\tau_l^s) \nonumber \\
    &=2^{ls(a - \delta)} \tr(\overline{\rho}_{l,\varepsilon} \tau_l^{\frac{s}{2}} 
	  \overline{\rho}_{l,\varepsilon}^{-s}\tau_l^{\frac{s}{2}}) \nonumber \\
    &=2^{ls(a - \delta)} \tr(\rho_{l,\varepsilon} \tau_l^{\frac{s}{2}} 
	  \overline{\rho}_{l,\varepsilon}^{-s}\tau_l^{\frac{s}{2}}). \label{comp_cq_hyp_ineq_2}
\end{align}
Here we used the fact, that $\overline{\rho}_{l,\varepsilon}$ and $\tau_l$ commute in the first equality. Eq. 
(\ref{comp_cq_hyp_ineq_2}) is justified, because the eigenprojections of $\tau_l$ wich appear in the definition 
of $\overline{\rho}_{l,\varepsilon}$ are absorbed by $\tau_l^{\frac{1}{2}}$. We can further upper bound the above 
expressions in the following way. Note, that
\begin{align} 
 \rho_{l,\varepsilon} \leq |\spec(\tau_l)| \overline{\rho}_{l,\varepsilon}. \label{comp_cq_hyp_spbd}
\end{align}
holds by Lemma \ref{fact-2}. Because $-(\cdot)^{-s}$ is an operator monotone function for every $s \in [0,1]$ 
(see e.g. \cite{bhatia97}), (\ref{comp_cq_hyp_spbd}) implies
\begin{align}
 \overline{\rho}_{l,\varepsilon}^{-s} \leq |\spec(\tau_l)|^s \rho_{l,\varepsilon}^{-s}. \nonumber
\end{align}
Using the above relation, one obtains
\begin{align}
\tr(\rho_{l,\varepsilon} \tau_l^{\frac{s}{2}}\overline{\rho}_{l,\varepsilon}^{-s}\tau_l^{\frac{s}{2}})
 &\leq |\spec(\tau_l)|^s \tr(\rho_{l,\varepsilon}\tau^{\frac{s}{2}}\rho_{l,\varepsilon}^{-s}\tau_l^{\frac{s}{2}}). \nonumber 
\end{align}
By combination with (\ref{comp_cq_hyp_ineq_2}) this leads to
\begin{align}
\tr((\eins_{\hr^{\otimes l}} - q_{l,\delta})\rho_{l,\varepsilon}) 
    & \leq |\spec(\tau_l)|^s 2^{ls(a-\delta)}\tr(\rho_{l,\varepsilon}\tau^{\frac{s}{2}}
	\rho_{l,\varepsilon}^{-s}\tau_l^{\frac{s}{2}}) \\
    & \leq (l+1)^{d^2} \exp \{l[(a-\delta)s - \tfrac{1}{l}\psi_{l,\varepsilon}(s)]\} \label{comp_cq_expo} \\
    & =  \exp \{l[(a-\delta)s - \tfrac{1}{l}\psi_{l,\varepsilon}(s) + w(l)]\}, \label{comp_cq_expo_2}
\end{align}
where $d:=\dim \hr$. In (\ref{comp_cq_expo}), we used the definition 
\begin{align}
\psi_{l,\varepsilon}(s) := -\log \tr(\rho_{l,\varepsilon}\tau_l^{\frac{s}{2}}\rho_{l,\varepsilon}^{-s}\tau_l^{\frac{s}{2}}),
\end{align}
in the last line we introduced the function $w$ defined by $w(l) := \frac{d^2}{l}\log(l+1)$ for every $l \in 
\nn$. Notice, that we also used the bound $|\spec(\tau_l)| \leq (l+1)^{d^2}$ on the spectrum of $\tau_l$ which is justified by Lemma \ref{fact-1}. In fact, by observation of (\ref{comp_cq_tau_l}), it is easy to see, that for every $\sigma$ in the tensor product representation of $S_l$ on $\hr^{\otimes l}$ (see (\ref{tensor_prod_rep})),
\begin{align}
(\eins_{\cc^{|\mathbf{X}|}}^{\otimes l} \otimes \sigma)\tau_l = \tau_l(\eins_{\cc^{|\mathbf{X}|}}^{\otimes l} \otimes \sigma)
\end{align}
holds. We will now show, that the argument of the exponential in (\ref{comp_cq_expo_2}) becomes strictly negative for a 
suitable choice of $s$, sufficiently small $\varepsilon$ and large enough $l$. We define
\begin{align}
 f_{l,\varepsilon}(s) := (a-\delta)s - \frac{1}{l}\psi_{l,\varepsilon}(s).
\end{align}
By the mean value theorem it suffices to show 
that $f'_{l,\varepsilon}(0) < 0$ for small enough $\varepsilon >0$.
For the derivative, we have
\begin{align}
 f'_{l,\varepsilon}(0) = a -\delta - \frac{1}{l}D(\rho_{l,\varepsilon}||\tau_l). \label{f_abl}
\end{align}
The relative entropy term in (\ref{f_abl}) can be lower bounded as follows. It holds
\begin{align}
 D(\rho_{l,\varepsilon}||\tau_l) 
    & = -S(\rho_{l,\varepsilon}) - \tr(\rho_{l,\varepsilon}\log \tau_l) \nonumber \\
    & = -S(\rho_{l,\varepsilon}) + lS(p) + S\left(\frac{1}{|T|}\sum_{t \in T} \sigma_t^{\otimes l}\right) \label{comp_cq_reg_ent} \\
    & \geq -S(\rho_{l,\varepsilon}) + l S(p) + \frac{1}{|T|} \sum_{t\in T} lS(\sigma_t). \label{d_low_bnd}
\end{align}
Notice that the equality in (\ref{comp_cq_reg_ent}) indeed holds, because the marginals on $(\cc^{|\mathbf{X}|})^{\otimes l}$ and $\hr^{\otimes l}$ of $\rho_l$ and $\tau_l$ are equal and therefore equal to the marginals of $\rho_{l,\epsilon}$ by definition for each $\epsilon \in (0,1)$. The inequality in (\ref{d_low_bnd}) is valid due to concavity of the von Neumann entropy. Because (\ref{fannes_abstand}) holds, 
\begin{align}
 S(\rho_{l,\varepsilon}) &\leq S(\rho_l) + 2\varepsilon \log \frac{\dim(\hr_l)}{2\varepsilon}  \nonumber \\
		      &\leq S(\rho_l) + 2\varepsilon l\log \frac{d}{2\varepsilon} \label{fannes_lwb}	  
\end{align}
is valid for $\epsilon < \frac{1}{2e}$, since for two states $\rho,\sigma \in \mathcal{S}(\hr)$ with $\|\rho-\sigma\|_1 \leq \varepsilon \leq \frac{1}{e}$, 
Fannes' inequality \cite{fannes73},
\begin{align}
 |S(\rho) - S(\sigma)| \leq \varepsilon \log\frac{\dim \hr}{\varepsilon},
\end{align}
is valid.
Together with (\ref{fannes_lwb}), (\ref{d_low_bnd}) implies
\begin{align}
 D(\rho_{l,\varepsilon}|| \tau_l) 
  & \geq -S(\rho_l) - 2\varepsilon l \log \frac{d}{2 \varepsilon} + l S(p) + \frac{1}{|T|} \sum_{t\in T} l S(\sigma_t) 
	\nonumber \\
  & \geq -\frac{1}{|T|} \sum_{t\in T} l S(\rho_t) - \log |T| - 2\varepsilon l \log \frac{d}{2 \varepsilon}
	    + l S(p) + \frac{1}{|T|} \sum_{t\in T} l S(\sigma_t) \label{lwb_2_5} \\
  & = \frac{l}{|T|} \sum_{t\in T} D(\rho_t || p \otimes \sigma_t) - \log|T|- 2\varepsilon l \log \frac{d}{2 \varepsilon}. \label{lwb_3}
\end{align}
The inequality in (\ref{lwb_2_5}) results from the fact, that the von Neumann entropy is an almost convex function, i.e.
\begin{align}   
S(\rho) \leq \sum_{i=1}^N {p_i} S(\rho_i) + \log(N)
\end{align}
for any mixture $\rho = \sum_{i=1}^N p_i \rho_i$ of states. Inserting (\ref{lwb_3}) in (\ref{f_abl}) gives
\begin{align}
f'_{l,\varepsilon}(0) 
  & \leq \min_{t \in T} D(\rho_t|| p \otimes \sigma_t) -\delta - 
    \frac{1}{|T|} \sum_{t\in T} D(\rho_t || p \otimes \sigma_t) + 2\varepsilon \log\frac{d}{2 \varepsilon}
    + \frac{1}{l} \log |T| \nonumber \\
  & < -\frac{\delta}{2} + \frac{1}{l} \log |T|, \label{abl_bnd}
\end{align}
provided that $0 < \varepsilon < \varepsilon_0(\delta)$ where $\varepsilon_0$ is small enough to ensure $2\varepsilon 
\log \frac{d}{2 \varepsilon} < \frac{\delta}{2}$. The mean value theorem shows that for $s \in (0,1]$
\begin{align}
 f_{l,\varepsilon}(s) = f_{l,\varepsilon}(0) + f'_{l,\varepsilon}(s')\cdot s \nonumber
\end{align}
holds for some $s' \in (0,s)$. Since $f_{l,\varepsilon}(0) = 0$, (\ref{abl_bnd}) shows that we can guarantee
\begin{align}
 f_{l,\varepsilon}(s) < \left(- \frac{\delta}{2}+ \frac{1}{l}\log |T|\right)s  \label{comp_cq_tr_ineq}
\end{align}
for small enough $s$. By (\ref{comp_cq_expo_2}) and (\ref{comp_cq_tr_ineq}) we obtain for $\varepsilon < \varepsilon_0(\delta)$
and $l$ large enough to make $w(l) < \frac{\delta s}{8}$ valid,
\begin{align}
 \tr((\eins_{\hr^{\otimes l}}-q_{l,\delta})\rho_{l,\varepsilon})  
    & \leq \exp \{l[f_{l,\varepsilon}(s) + w(l)]\} \nonumber \\
    & \leq \exp \left\{-l\left(\frac{\delta}{4}s - w(l)\right)\right\} \\
    & \leq |T| \cdot \exp \{-l\frac{\delta}{8}s\}. \nonumber
\end{align}
Using (\ref{fannes_abstand}), we have (with $\varepsilon < \varepsilon_0$) 
\begin{align}
 \tr((\eins_{\hr^{\otimes l}} - q_{l,\delta})\rho_{l}) 
    & \leq \|\rho_{l,\varepsilon} - \rho_l \|_1 + \tr((\eins_{\hr^{\otimes l}} - q_{l,\delta})\rho_{l,\varepsilon}) \nonumber \\
    & \leq 2 \varepsilon + |T| \exp\{-l \frac{\delta}{8}s\}. \nonumber
\end{align}
We can in fact, choose the parameter $\epsilon$ dependent on $l$ in a way that $(\varepsilon_l)_{l=1}^\infty$ decreases exponentially in $l$, which proves the second claim of the lemma.\qed
\end{proof}
In order to prove the direct part of the coding theorem for general sets of channels we have to approximate arbitrary sets of channels by finite ones. 
For $\alpha>0$, an $\alpha$-net in $CQ(\mathbf{X},\hr)$ is a finite set $\mathcal{N}_{\alpha} := \{W_i\}_{i=1}^{N_\alpha}
\subset CQ(\mathbf{X}, \hr)$ with the property, that for every channel $W \in CQ(\mathbf{X}, \hr)$ there exists an index 
$i \in [N_\alpha]$ such that
\begin{align}
 \|W - W_i\|_{cq} < \alpha
\end{align}
holds. For a given set $\mathcal{W} \subset CQ(\mathbf{X},\hr)$ an $\alpha$-net $\mathcal{N}_{\alpha}$ in $CQ(\mathbf{X}, \hr)$ generates an approximating set $\widetilde{\mathcal{W}}_{\alpha}$ defined by
\begin{align}
 \widetilde{\mathcal{W}}_{\alpha} := \{W_i \in \mathcal{N}_\alpha : B_{cq}(\alpha, W_i) \cap \mathcal{W} \neq \emptyset\}.
\end{align}
where $B_{cq}(\alpha, A)$ denotes the $\alpha$-ball with center $A$ regarding the norm $\|\cdot\|_{cq}$. The above definition does not guarantee, that 
$\widetilde{\mathcal{W}}_{\alpha}$ is a subset of $\mathcal{W}$ but each $\widetilde{\mathcal{W}}_\alpha$ clearly generates a set $\mathcal{W}_{2\alpha}\subset\mathcal W$ of at most the same 
cardinality, such that for every $W \in \mathcal{W}$ exists an index $i \in [N_{\alpha}]$ 
with
\begin{align}
 \|W - W_i\|_{cq} < 2\alpha.
\end{align}
The next lemma states that we find good approximations of arbitrary compound cq-channels among such sets as defined above. The proof can be given by minor variations of the corresponding results in \cite{bb-compound}, \cite{bbn-1}, and we omit it here.
\begin{lemma}\label{lemma:approximation}
Let $\mathcal{W} := \{W_t\}_{t \in T} \subset CQ(\mathbf{X},\hr)$ and $\alpha \in (0,\frac{1}{e})$. There exists a set 
$T_\alpha \subseteq T$ which fulfills the following conditions
\begin{enumerate}
 \item $|T_{\alpha}| < \left(\frac{6}{\alpha}\right)^{2|\mathbf{X}|d^2}$,
 \item given any $l \in \nn$, to every $t \in T$ one finds an index $t' \in T_{\alpha}$ such that 
	\begin{align}	
	  \|W_t^{\otimes l}(x^l) - W_{t'}^{\otimes l}(x^l)\|_1 < 2l\alpha.
	\end{align}
       holds for every $x^l \in \mathbf{X}^l$. Moreover,
 \item for every $p \in \mathfrak{P}(\mathbf{X})$, 
       \begin{align}
	\left|\underset{t'\in T_{\alpha}}{\min} \chi(p,W_{t'}) - \underset{t \in T}{\inf} \chi(p, W_t)\right| 
	  \leq 2\alpha \log\frac{d}{2\alpha}
       \end{align}
	holds.
\end{enumerate}
\end{lemma}
The following lemma is from \cite{ahlswede69} and will be used to establish the weak converse in Theorem \ref{comp_cq_direct_part}. It states that codes which have small 
average error probability for a finite compound cq-channel contain subcodes with good maximal error probability of not substantially smaller size. 
\begin{lemma}[cf. \label{average_to_maximum}\cite{ahlswede69}, Lemma 1]\label{lemma7}
Let $\mathcal{W} :\{W_t\}_{t \in T} \subset CQ(\mathbf{X},\hr)$ be a compound channel with $|T| < \infty$ and $l \in 
\nn$. If $(u_i^l,D_i^l)_{i=1}^{M_l}$ is an $(l,M_l)$-code with
\begin{align}
 \max_{t\in T}\ \frac{1}{M_l}\sum_{i=1}^{M_l}\tr(W_t^{\otimes l}(u_i^l)(\eins_\hr - D_i^l)) \leq \overline{\lambda}.
\end{align}
Then there exists for every $\epsilon>0$ a subcode $(u_{i_j}^l,D_{i_j}^l)_{j=1}^{M_{l,\epsilon}}$ of size 
$M_{l,\epsilon} = \lfloor \frac{\epsilon}{1 - \epsilon} M_l \rfloor$ with
\begin{align}
\max_{t \in T}\ \max_{j \in [M_{l,\epsilon}]} \tr(W_t^{\otimes l}(u_{i_j}^l)(\eins_\hr - D_{i_j}^l)) \leq |T| (\overline{\lambda} + \epsilon)
\end{align}
\end{lemma}
Finally, we have gathered all the prerequisites to prove Theorem \ref{comp_cq_direct_part}:
\begin{proof}[of Theorem \ref{comp_cq_direct_part}]
The direct part (i.e. the assertion that the r.h.s. lower-bounds the l.h.s. in (\ref{comp_cq_capacity})) is proven by combining Lemma \ref{lemma:statistical-test-to-code} with Lemma \ref{lemma:existence-of-statistical-test}. 
Let $p = \argmax_{p' \in \pr(\mathbf{X})} \inf_{t \in T} \chi(p', W_t)$. We show that for any $\delta > 0$, 
\begin{align}
 \inf_{t \in T} \chi(p,W_t) - \delta \label{comp_cq_direct_rate}
\end{align}
is an achievable rate. We can restrict ourselves to the case, where $\inf_{t \in T} \chi(p,W_t) >  \delta > 0$ holds,  because otherwise the above statement is trivially fulfilled. The above mentioned lemmata consider finite sets of channels, therefore we choose an approximating set $\mathcal{W}_{\alpha_l}$ (of cardinality $T_{\alpha_{l}}$) according to Lemma \ref{lemma:approximation} for every $l \in \nn$, where we leave the sequence $\alpha_1,\alpha_2,...$ initially unspecified. 
For every $l \in \nn$ and $t' \in T_{\alpha_l}$, let $\rho_{t'}$, $\sigma_{t'}$ be defined according to eq. 
(\ref{comp_cq_rho_t}) and (\ref{comp_cq_sigma_t}), and further define states
\begin{align}
 \rho_l := \frac{1}{|T_{\alpha_l}|} \sum_{t' \in T_{\alpha_l}} v_l\rho_{t'}^{\otimes l}v_l^\ast
\end{align}
and
\begin{align}
 \tau_l := p^{\otimes l} \otimes \frac{1}{|T_{\alpha_l}|} \sum_{t' \in T_{\alpha_l}} \sigma_{t'}^{\otimes l}.
\end{align}
For a given number $\eta$ with $0 < \eta < a_l$, Lemma \ref{lemma:existence-of-statistical-test} guarantees (for large enough $l$), with a suitable constant
$\tilde{c}>0$, the existence of a projection $q_{l,\eta} \in 
\bo((\cc^{|\mathbf{X}|})^{\otimes l} \otimes \hr^{\otimes l})$ with
\begin{align}
\tr(q_{l,\eta}\rho_l) \geq 1 - |T_{\alpha_l}| \cdot 2^{-l\tilde{c}}
\end{align}
and
\begin{align}
 \tr(q_{l,\eta}\tau_l) \leq 2^{-l(a_l-\eta)}
\end{align}
where we defined $a_l := \min_{t' \in T_{\alpha_l}} D(\rho_{t'} ||p \otimes \sigma_{t'})$. 
This by virtue of Lemma \ref{lemma:statistical-test-to-code} implies for every $\gamma > 0$ such that $\eta + \gamma < a_l$ the existence of a cq-code 
$(x_m^{l},D_m^l)_{m \in [M_l]}$ of size 
\begin{align}
 M_l = \lfloor 2^{l(a_l - \eta - \gamma)}\rfloor \label{comp_dir_rte}
\end{align}
and average error bounded by
\begin{align}
 \max_{t' \in T_{\alpha_l}}\ \frac{1}{M_l}\sum_{m=1}^{M_l}\tr(W_{t'}^{\otimes l}(u_m^l)(\eins_{\hr^{\otimes l}}-D_m^l))		
    \leq 2 |T_{\alpha_l}|^2 2^{-l\tilde{c}} + 4\cdot |T_{\alpha_l}|2^{-l\gamma} \label{comp_cq_dir_error_1}.
\end{align}
Notice, that for other positive numbers $\gamma, \delta$, trivial codes have $M_l = 1  \geq \lfloor 2^{l(a_l - \eta - \gamma)}\rfloor$.
Using (\ref{comp_dir_rte}) we obtain, 
\begin{align}
 \frac{1}{l} \log M_l & \geq  \min_{t' \in T_{\alpha_l}} \chi(p, W_{t'})- \eta - \gamma\\
                      & \geq  \inf_{t \in T} \chi(p, W_t) - \eta -\gamma - 4 \alpha_l \log\frac{d}{2 \alpha_l}, \label{endrate}
\end{align}
where  the second inequality follows from Lemma \ref{lemma:approximation}. For the average error, it holds, 
\begin{align}
  &\sup_{t \in T}\frac{1}{M_l} \sum_{m \in [M_l]} \tr\left(W_t^{\otimes l}(u_m^l)(\eins_{\hr^{\otimes l}} - D_m^l)\right) \\
 &\leq \max_{t' \in T_{\alpha_l}} \frac{1}{M_l} \sum_{m \in [M_l]} \tr\left(W_{t'}^{\otimes l}(u_m^l)(\eins_{\hr^{\otimes l}} - D_m^l)\right) + 2l\alpha_l \\
 &\leq 2|T_{\alpha_l}|^22^{-l\tilde{c}} + 4|T_{\alpha_l}|2^{-l\gamma} + 2l\alpha_l. \label{endfehler}
\end{align}
The first of the above inequalities follows from Lemma \ref{lemma:approximation}, the second one is by 
(\ref{comp_cq_dir_error_1}). Because we chose the approximating sets according to Lemma \ref{lemma:approximation}, 
\begin{align}
|T_{\alpha_l}| \leq \left(\frac{6}{\alpha_l}\right)^{2|\mathbf{X}|d^2}
\end{align}
holds. In fact, if we specify $\alpha_l$ to be $\alpha_l := 2^{-l\hat{c}}$ for every $l \in \nn$, where $\hat{c}$ is a 
constant with $0 < \hat{c} < \min\left\{\frac{\tilde{c}}{4|\mathbf{X}|d^2},\frac{\eta}{2|\mathbf{X}|d^2} \right\}$ , 
the r.h.s of (\ref{endfehler}) decreases exponentially for $l \rightarrow \infty$.
If we additionally choose $\eta$ and $\gamma$, small enough to validate $\delta > \eta + \gamma + 2 \alpha_l \log \frac{d}{2\alpha_l}$
for sufficiently large $l$, the rate defined in (\ref{comp_cq_direct_rate}) is shown to be achievable by (\ref{endfehler}) and (\ref{endrate}). Since $\delta$ was arbitrary, the direct statement follows.\\
It remains to prove the converse statement. For the proof, we will construct a good code for transmission under the maximal error criterion and invoke the strong converse result given in \cite{bb-compound} (see Remark \ref{str_con_rem_1}). We show, that for any $\delta > 0$,
\begin{align} 
 \overline{C}_C(\mathcal{W}) < \max_{p \in \pr(\mathbf{X})}\ \inf_{t \in T}\ \chi(p,W_t) + \delta. \label{comp_cq_conv_r}
\end{align}
Let $\delta > 0$ and assume that for some fixed $l \in \nn$, $\mathcal{C}_l:=(u_m^l, D_m^l)_{m = 1}^{M_l}$ is an $(l,M_l)$-code with
\begin{align}
 \sup_{t \in T}\ \frac{1}{M_l}\sum_{m=1}^{M_l}\tr(W_t^{\otimes l}(u_m^l)(\eins_{\hr^{\otimes l}}-D_m^l)) \leq \overline{\lambda}_l.
\end{align}
We always can find a finite subset $\hat{T} \subset T$ such that
\begin{align}
 \left|\underset{p \in \pr(\mathbf{X})}{\max} \underset{t \in T}{\inf} \chi(p,W_t) -  
	  \underset{p \in \pr(\mathbf{X})}{\max} \underset{t \in \hat{T}}{\min} \chi(p,W_t) \right| \leq \frac{\delta}{2}
       \label{chi_nahe_dran}
\end{align}
holds (e.g. a set $T_{\alpha}$ as in Lemma \ref{lemma:approximation} for suitable $\alpha$). We set $\epsilon := \frac{1}{2|\hat{T}|}$. By virtue of Lemma \ref{average_to_maximum} we find a subcode 
$(u_{i_j}^l, D_{i_j}^l)_{j=1}^{M_{l,\epsilon}} \subseteq \mathcal{C}_l$ of $\mathcal{C}_l$ which has 
size 
\begin{align}
 M_{l,\epsilon} := \left\lfloor \frac{\epsilon}{1 - \epsilon} M_l \right\rfloor \label{comp_cq_str_conv_rate}
\end{align}
and maximal error bounded by
\begin{align}
 \max_{t \in \hat{T}}\ \max_{j \in M_{l, \epsilon}} \tr\left(W_t^{\otimes l}(u_{i_j}^l)(\eins_{\hr^{\otimes l}} - D_{i_j}^l)\right) 
  \leq \overline{\lambda}_l |\hat{T}| + \frac{1}{2}. \\
\end{align}
If $l$ is sufficiently large, the r.h.s. is strictly smaller than one. Therefore, by the strong converse theorem for 
coding under the maximal error criterion (see \cite{bb-compound}, Theorem 5.13), we have (with some constant $K>0$)
\begin{align}
 \frac{1}{l}\log M_{l,\epsilon} 
   &\leq \underset{p \in \pr(\mathbf{X})}{\max} \underset{t \in \hat{T}}{\min} \chi(p,W_t) + K \frac{1}{\sqrt{l}} \\
   &\leq \underset{p \in \pr(\mathbf{X})}{\max} \underset{t \in T}{\inf} \chi(p,W_t) + 
      \frac{\delta}{2} + K \frac{1}{\sqrt{l}}. \label{comp_cq_conv_contr_1}
\end{align}
The second line above follows from (\ref{chi_nahe_dran}). On the other hand, by (\ref{comp_cq_str_conv_rate}), we have
\begin{align}
 \log M_l \leq \log M_{l,\epsilon} + \log\left(\frac{\varepsilon}{2(1-\varepsilon)}\right). \label{comp_cq_conv_contr_2}
\end{align}
Dividing both sides of (\ref{comp_cq_conv_contr_2}) by $l$ and combinig the result with (\ref{comp_cq_conv_contr_1}) shows that for sufficiently large $l$
\begin{align}
 \frac{1}{l} \log M_l 
  &\leq \underset{p \in \pr(\mathbf{X})}{\max} \underset{t \in T}{\inf} \chi(p,W_t) + 
      \frac{\delta}{2} + K \frac{1}{\sqrt{l}} + \frac{1}{l}\log\left(\frac{\varepsilon}{2(1-\varepsilon)}\right) \\
  &\leq \underset{p \in \pr(\mathbf{X})}{\max} \underset{t \in T}{\inf} \chi(p,W_t) + \delta
\end{align}
holds, which shows (\ref{comp_cq_conv_r}). Since $\delta$ was an arbitrary positive number, we are done.\qed
\end{proof}
\end{section}
\begin{remark}\label{str_con_rem_1}
While the achievability part for cq-compound channels regarding the maximal error criterion given in \cite{bb-compound} required technical effort, the strong converse proof 
was rather uncomplicated. It was given there by a combination of Wolfowitz' proof technique for the strong converse in case of classical compound channels and a lemma from 
\cite{winter99}.
\end{remark}
\begin{remark}\label{remark10}
We remark here, that a general strong converse does not hold for the capacity of compound cq-channels if the average error
is considered as criterion for reliability of the message transmission. This can be seen by a counterexample given by Ahlswede in \cite{ahlswede68} (Example 1) regarding 
classical compound channels. However, we will see in the proof of Theorem \ref{avcqc_strong_converse}, that in certain situations (especially, where $\mathcal{W}$ is a 
convex set) a strong converse proof can be established. 
\end{remark}
As a corollary to the achievability part of Theorem \ref{comp_cq_direct_part} above, we immediately obtain a direct coding theorem for the capacity of a finite cq-compound channel under the maximal error criterion. 
\begin{corollary}
For a finite compound cq-channel $\mathcal{W}:=\{W_t\}_{t \in T} \subset CQ(\mathbf{X},\hr)$ we have
\begin{align}
 C_C(\mathcal{W}) \geq \underset{p \in \pr(\mathbf{X})}{\max}\, \underset{t \in T}{\min} \, \chi(p,W_t)
\end{align}
\end{corollary}
\begin{proof}
For an arbitrary number $\delta > 0$, we show, that 
\begin{align}
 \underset{p \in \pr(\mathbf{X})}{\max}\, \underset{t \in T}{\min} \, \chi(p,W_t) - \delta
\end{align}
is an achievable rate. Let $\{\mathcal{C}_l\}_{l \in \nn}$, $\mathcal{C}_l := (u_m^l,D_m^l)_{m=1}^{M_l} \forall l \in \nn$,
 be a sequence of $(l,M_l)$-codes with 
\begin{align}
 \liminf_{l \to \infty} \frac{1}{l}\log M_l \geq \underset{p \in \pr(\mathbf{X})}{\max}\, \underset{t \in T}{\min} \, \chi(p,W_t) - \frac{1}{\delta}.
\end{align}
and
\begin{align}
 \max_{t \in T} \frac{1}{M_l} \sum_{m=1}^{M_l} \tr\left(W_t^{\otimes l}(u_m^l)(\eins_{\hr^{\otimes l}}-D_m^l)\right) \leq \lambda_l
\end{align}
for every $l \in \nn$, where $\lim_{l \rightarrow \infty} \lambda_l = 0$. Such codes exist by virtue of Theorem \ref{comp_cq_direct_part}. 
Because of Lemma \ref{average_to_maximum}, we find for each $l \in \nn$  a subcode 
$\tilde{\mathcal{C}}_l := (u_{m_i}^l,D_{m_i}^l)_{i \in [\tilde{M}_l]} \subseteq \mathcal{C}_l$ of size
 $\widetilde{M}_l := \lfloor \frac{\epsilon_l}{1 - \epsilon_l}M_l \rfloor$ and maximal error
 \begin{align}
  \max_{t \in T}\ \max_{i \in [\widetilde{M}_l]} \tr\left(W_t^{\otimes l}(u_{m_i}^l)(\eins_{\hr^{\otimes l}}-D_{m_i}^l)\right) \leq (\lambda_l + \epsilon_l)|T|. \label{comp_cq_corr_1}
 \end{align}
with the sequence $(\epsilon_l)_{l=1}^\infty$ defined by $\epsilon_l := 2^{-l\frac{\delta}{3}}$ f.a. $l \in \nn$,
 it is clear that we find a sequence of $(l,\tilde{M}_l)$-subcodes $\{\tilde{\mathcal{C}}_l\}_{l\in \nn}$, 
where $\widetilde{\mathcal{C}}_l := (u_{m_i}^l,D_{m_i}^l)_{i=1}^{\widetilde{M}_l}$ f.a. $l \in \nn$, which fulfills  
 \begin{align}
  \lim_{l \rightarrow \infty}  \max_{t \in T}\ \max_{i \in [\widetilde{M}_l]} \tr\left(W_t^{\otimes l}(u_{m_i}^l)
(\eins_{\hr^{\otimes l}}-D_{m_i}^l)\right) = 0
 \end{align}
 and 
 \begin{align}
  \underset{l \rightarrow \infty}{\liminf}\frac{1}{l} \log \tilde{M}_l = \liminf_{l \to \infty}\frac{1}{l}\log M_l \geq \underset{p \in \pr(\mathbf{X})}{\max}\, \underset{t \in T}{\min} \, \chi(p,W_t) - \delta.  
 \end{align}\qed
\end{proof}
\begin{remark}
 The above corollary, although proven here for finite sets, can be extended to arbitrary compound sets by approximation 
 arguments, as carried out in \cite{bb-compound}. Moreover, an inspection of the proofs in this section shows that the speed of convergence of the errors
remains exponential.
\end{remark}
\begin{section}{\label{sec:avcqc}AVCQC}
\begin{subsection}{\label{subsec:The Ahlswede-Dichotomy for AVcqCs}The Ahlswede-Dichotomy for AVcqCs}
In this section, we prove Theorem \ref{ahlswede-dichotomy-for-avcqcs} and Theorem \ref{avcqc_strong_converse}. The proof of Theorem \ref{ahlswede-dichotomy-for-avcqcs} is 
carried out via robustification of codes for a suitably chosen 
compound cq-channel. 
More specifically, to a given AVcqC $\A$ we take a sequence of codes for the compound channel $\W:=\conv(\A)$ that operates close to the capacity of
$\W$. Thanks to Theorem 
\ref{comp_cq_direct_part}, we know that there exist codes for $\W$ that, additionally, have an exponentially fast decrease of average error
probability. The robustification 
technique then produces a sequence of random codes for $\A$ that have a discrete, but super-exponentially large support and, again, an exponentially
fast decrease of average 
error probability.\\
An intermediate result here is the (tight) lower bound on $\overline C_{A,r}(\A)$.\\
A variant of the elimination technique of \cite{ahlswede-elimination} is proven that is adapted to AVcqCs and reduces the amount of randomness from
super-exponential to polynomial, 
while slowing down the speed of convergence of the average error probability from exponential to polynomial at the same time.\\
Then, under the assumption that $C_{A,d}(\A)>0$ holds, the sender can send the required amount of subexponentially many messages in order to establish
sufficiently much common 
randomness. After that, sender and receiver simply use the random code for $\A$.\\
We now start out on our predescribed way. The following Theorem \ref{robustification-technique} and Lemma \ref{innerproduct-lemma} will be put to good
use, but are far from 
being new so we simply state them without proof.\\
Let, for each $l\in\nn$, $\textrm{Perm}_l$ denote the set of permutations acting on $\{1,\ldots, l\}$. Let us further suppose that we are given a
finite set $\bS$. We use the 
natural action of $\textrm{Perm}_l$ on $\bS^l$ given by $\sigma:\mathbf S^l\rightarrow\mathbf S^l$, $\sigma(s^l)_i:=s_{\sigma^{-1}(i)}$.\\ 
Let $T(l,\bS)$ denote the set of types on $\bS$ induced by the elements of $\bS^l$, i.e. the set of empirical distributions on $\bS$ generated by
sequences in 
$\bS^l$. Then Ahlswede's robustification can be stated as follows.
\begin{theorem}[Robustification technique, cf. Theorem 6 in \cite{ahlswede-coloring}]\label{robustification-technique}\ \\
Let $\bS$ be a set with $|\bS|<\infty$ and $l\in\nn$. If a function $f:\bS^l\to [0,1]$ satisfies
\begin{equation}\label{eq:robustification-1}
 \sum_{s^l\in\bS^l}f(s^l)q(s_1)\cdot\ldots\cdot q(s_l)\ge 1-\gamma
\end{equation}
for all $q\in T(l,\bS)$ and some $\gamma\in [0,1]$, then
\begin{equation}\label{eq:robustification-2}
  \frac{1}{l!}\sum_{\sigma\in\textup{Perm}_l}f(\sigma(s^l))\ge 1-(l+1)^{|\bS  |}\cdot \gamma\qquad \forall s^l\in \bS^l.
\end{equation}
\end{theorem}
The original theorem can, together with its proof, be found in \cite{ahlswede-coloring}. A proof of Theorem \ref{robustification-technique} can be
found in \cite{abbn}. The following 
Lemma is borrowed from \cite{ahlswede-elimination}.
\begin{lemma}\label{innerproduct-lemma}
Let $K\in\nn$ and real numbers $a_1,\ldots,a_K,b_1,\ldots, b_K\in [0,1]$ be given. Assume that
\begin{equation}\frac{1}{K}\sum_{i=1}^Ka_i\ge 1-\eps\qquad \textrm{and} \qquad \frac{1}{K}\sum_{i=1}^K b_i\ge 1-\eps,  \end{equation}
hold. Then
\begin{equation}\frac{1}{K}\sum_{i=1}^Ka_ib_i\ge 1-2\eps.  \end{equation}
\end{lemma}
We now come to the promised application of the robustification technique to AVcqCs.
\begin{lemma}\label{lemma:avcqc-random-achiev}
Let $\A=\{A_s \}_{s\in\bS}$ be an AVcqC. For every $\eta>0$ there is a sequence of $(l,M_l)$-codes for the compound channel $\W:=\conv(\A)$ and an 
$l_0\in\nn$ such that the following two statements are true.
\begin{equation}\label{conversion-1}
\liminf_{l\rightarrow\infty}\frac{1}{l}\log M_l\ge \overline{C}_{\textup C}(\W)  -\eta
\end{equation}
\begin{equation}\label{conversion-2}
 \min_{s^l\in\bS^l}\frac{1}{l!}\sum_{\sigma\in\textup{Perm}_l}\frac{1}{M_l}\sum_{i=1}^{M_l}\tr(A_{s^l}(\sigma^{-1}(x^l_i))\sigma^{-1}(D_i^l))\ge 1-
(l+1)^{|\bS|}\cdot 2^{-lc}\qquad \forall l\geq l_0
\end{equation}
with a positive number $c=c(|\mathbf X|,\dim\hr,\A,\eta)$.
\end{lemma}
\begin{remark}
The above result can be gained for arbitrary, non-finite sets $\bS$ as well. A central idea then is the approximation of $\conv(\A)$ from the outside
by a convex polytope. 
Since $CQ(\mathbf X,\hr)$ is not a polytope itself (except for trivial cases), an additional step consists of applying a depolarizing channel $\cn_p$
and approximate 
$\cn_p(\conv(\A))$, a set which does not touch the boundary of $CQ(\mathbf X,\hr)$, instead of $\conv(\A)$.\\ 
This step can then be absorbed into the measurement operators, i.e. one uses operators $\cn_p^*(D_i^l)$ instead of the original $D_i^l$
($i=1,\ldots,M_l$).\\
A thorough application of this idea can be found in \cite{abbn}, where the robustification technique gets applied in the case of entanglement
transmission over arbitrarily varying 
quantum channels.
\end{remark}
\begin{proof}
According to Lemma \ref{comp_cq_direct_part} there is a sequence of $(l,M_l)$ codes for the compound channel
$\conv(\A)=\{W_q:W_q=\sum_{s\in\bS}q(s)A_s,\ q\in\mathfrak P(\bS)\}$ fulfilling
\begin{eqnarray}
\liminf_{l\rightarrow\infty}\frac{1}{l}\log M_l\geq\overline{C}_{\textup C}(\conv(\A))  -\eta\end{eqnarray}
and
\begin{eqnarray}\exists l_0\in\nn:\
\inf_{W\in\conv(\A)}\frac{1}{M_l}\sum_{i=1}^{M_l}\tr(W^{\otimes l}(x^l_i)D_i^l)\geq1-2^{-lc}\ \forall l\geq l_0.
\end{eqnarray}
The idea is to apply Theorem \ref{robustification-technique}. Let us, for the moment, fix an $\nn\ni l\geq l_0$ and define a function
$f_l:\bS^l\rightarrow[0,1]$ by
\begin{equation}
 f_l(s^l):=\frac{1}{M_l}\sum_{i=1}^{M_l}\tr(A_{s^l}(x^l_i)D_i^l). \label{avcqc_achiev_funct}
\end{equation}
Then for every $q\in\mathfrak P(\bS)$ we have
\begin{equation}
\sum_{s^l\in\bS^l}f_l(s^l)\prod_{i=1}^lq(s_i)=\frac{1}{M_l}\sum_{i=1}^{M_l}\tr(W_q^{\otimes l}(x^l_i)D_i^l)\geq1-2^{-lc}. \label{avcqc_ach_av_to_convex}
\end{equation}
It follows from Theorem \ref{robustification-technique}, that 
\begin{align}
1-(l+1)^{|\bS  |}\cdot 2^{-lc}
& \leq\frac{1}{l!}\sum_{\sigma\in\textup{Perm}_l}f_l(\sigma(s^l)) \\
& =\frac{1}{l!}\sum_{\sigma\in\textup{Perm}_l}\frac{1}{M_l}\sum_{i=1}^{M_l}\tr(A_{
s^l}(\sigma^{-1}(x^l_{i}))\sigma^{-1}(D_{i}^l))\qquad \forall s^l\in \bS^l
\end{align}
holds, where
\begin{equation}
\sigma(B_1\otimes\ldots\otimes B_l):=B_{\sigma^{-1}(1)}\otimes\ldots\otimes B_{\sigma^{-1}(l)}\qquad \forall\ B_1,\ldots,B_l\in\B(\hr) 
\end{equation} 
defines, by linear extension, the usual representation of $\textrm{Perm}_l$ on $\mathcal B(\hr)^{\otimes l}$ and the action of $\textrm{Perm}_l$ on $\bX^l$ is analogous to that on $\bS^l$.\qed
\end{proof}
It is easily seen from the above Lemma \ref{lemma:avcqc-random-achiev} and Theorem \ref{comp_cq_direct_part}, that the following theorem holds.
\begin{theorem}
For every AVcqC $\A$, 
\begin{align}
\overline C_{A,r}(\A)\geq\overline C_{C}(\conv(\A))=\max_{p\in\pr(\bX)}\inf_{A\in\conv(\A)}\chi(p,A). \label{avcqc_rand_dir}
\end{align}
\end{theorem}
In the following we give a proof of the remaining inequality in (\ref{ahlswede-dichotomy-for-avcqcs_1}). In fact, we  
prove the stronger statement Theorem \ref{avcqc_strong_converse}:
\begin{proof}[of Theorem \ref{avcqc_strong_converse}:]
We define $\mathcal{W}:= \conv(\mathcal{A})$. Since $|\bS|$ is finite, this set is compact. 
The function $\chi(\cdot,\cdot)$ is a concave-convex function (see eq. (\ref{def:holevo-chi})), therefore by the Minimax Theorem, 
\begin{align}
    \underset{p \in \pr(\mathbf{X})}{\max}\ \underset{W \in {\mathcal{W}}}{\min} \chi(p,W)
  = \underset{W \in {\mathcal{W}}}{\min}\ \underset{p \in \pr(\mathbf{X})}{\max} \chi(p,W) \label{avcqc_minimax}
\end{align}
holds. Both sides of the equality are well defined, because we are dealing with a compact set. Let an arbitrary $W_q \in \mathcal{W}$ be given by the formula
\begin{align}
 W_q = \sum_{s\in\bS} q(s) A_{s},
\end{align}
where $q\in\pr(\bX)$. Set, for every $l\in\nn$, $q^{\otimes l}(s^l):=\prod_{i=1}^lq(s_i)$. Let $\lambda\in[0,1)$, $\delta>0$ and $(\mu_l)_{l\in\nn}$ be a sequence of $(l,M_l)$-random codes such 
that both
\begin{align}
\liminf_{l\to\infty}\frac{1}{l}\log M_l= \overline C_{A,r}(\A,\lambda)-\delta
\end{align}
and
\begin{align}
\liminf_{l\to\infty}\min_{s^l\in\bS^l}\frac{1}{M_l}\sum_{i=1}^{M_l}\tr(A_{s^l}(u^l_i)D^l_i)d\mu_l((u^l_i,D^l_i)_{i=1}^{M_l})\geq1-\lambda.
\end{align}
For every $l\in\nn$ it holds that
\begin{align}
&\int \sum_{i=1}^{M_l} \tr(W_q^{\otimes l}(u_i^l)D_i^l)\ d\mu_l((u_i^l,D_i^l)_{i=1}^{M_l}) \\
&=\sum_{s^l \in \bS^l}q^{\otimes l}(s^l) \int \sum_{i=1}^{M_l} \tr(A_{s^l}(u_i^l)D_i^l)\ d\mu_l((u_i^l,D_i^l)_{i=1}^{M_l}) \label{avcqc_str_con_1}\\ 
&\geq\min_{s^l\in\bS^l}\int\sum_{i=1}^{M_l} \tr(A_{s^l}(u_i^l)D_i^l)\ d\mu_l((u_i^l,D_i^l)_{i=1}^{M_l}),
\end{align}
which shows, that 
\begin{align}
\liminf_{l\to\infty}\int\frac{1}{M_l}\sum_{i=1}^{M_l} \tr(W_q^{\otimes l}(u_i^l)\D_i^l)\ d\mu_l((u_i^l,D_i^l)_{i=1}^{M_l})\geq1-\lambda
\end{align}
holds. It follows the existence of a sequence $(u^l_i,D^l_i)_{l\in\nn}$ of $(l,M_l)$-codes for the discrete memoryless cq-channel $W_q$ satisfying
\begin{align}
&\liminf_{l\to\infty}\frac{1}{l}\log M_l=\overline C_{A,d}(\A,\lambda)-\delta\qquad \textrm{and} \\
&\liminf_{l\to\infty}\frac{1}{M_l}\sum_{i=1}^{M_l}\tr(W_q^{\otimes l}(u^l_i)D^l_i)\geq1-\lambda. 
\end{align}
By virtue of the strong converse theorem for single cq-DMCs given in \cite{winter99} (also to be found and independently obtained in \cite{ogawa-nagaoka}), for any $\lambda \in [0,1)$, $\delta>0$ it follows
\begin{align}
\overline C_{A,r}(\A,\lambda)-\delta&=\liminf_{l\to\infty}\frac{1}{l} \log M_l\\
&\leq \underset{p \in \pr(\mathbf{X})}{\max}(p,W_q)
\end{align}
and, since $W_q\in\mathcal W$ was arbitrary,
\begin{align}
\overline C_{A,r}(\A,\lambda)-\delta&\leq \min_{W \in {\mathcal{W}}}\underset{p \in \pr(\mathbf{X})}{\max}\chi(p,W) \label{avcqc_str_con_2} \\
&= \underset{p \in \pr(\mathbf{X})}{\max}\min_{W \in \mathcal{W}}\chi(p,W). \label{avcqc_str_con_3}
\end{align}
The equality in (\ref{avcqc_str_con_3}) holds by 
(\ref{avcqc_minimax}). Since $\delta$ was an arbitrary positive number, we are done.\qed
\end{proof}
The following lemma contains the essence of the derandomization procedure.
\begin{lemma}[Random Code Reduction]\label{random-code-reduction}
Let $\A=\{A_s\}_{s\in\mathbf S}$ be an AVcqC, $l\in\nn$, $\mu_l$ an $(l,M_l)$ random code for $\A$ and $1>\eps_l\geq0$ with
\begin{equation}\label{eq:random-code-reduction}
e(\mu_l,\A):=\inf_{s^l\in\bS^l}\int \frac{1}{M_l}\sum_{i=1}^{M_l}\tr(A_{s^l}(x^l_i)D_i^l)d\mu_l((x^l_i,D_i^l)_{i=1}^{M_l})\ge
1-\eps_l.
\end{equation}
Let $n,m\in\mathbb R$. Then if $4\eps_l\leq l^{-m}$ and $2\log|\bS|< l^{n-m-1}$ there exist $l^n$ $(l,M_l)$-deterministic codes 
$(x^l_{1,j},\ldots,x^l_{M_l,j},D^l_{1,j},\ldots,D^l_{M_l,j})$ ($1\leq j\leq l^n$) for $\A$  such that
\begin{equation}\label{eq:random-code-reduction-a}
\frac{1}{l^n}\sum_{j=1}^{l^n} \frac{1}{M_l}\sum_{i=1}^{M_l}\tr(A_{s^l}(x^l_{i,j})D^l_{i,j})\geq1-l^{-m} \qquad  \forall s^l\in\mathbf S^l.
\end{equation}
\end{lemma}
\begin{proof}
Set $\eps:=2l^{-m}$. By the assumptions of the lemma we have
\begin{equation}
e(\mu_l,\A):=\min_{s^l\in\bS^l}\int \frac{1}{M_l}\sum_{i=1}^{M_l}\tr(A_{s^l}(x^l_i)D_i^l)d\mu_l((x_i^l,D_i^l)_{i=1}^{M_l})\ge
1-\eps_l.
\end{equation}
For a fixed $K\in\nn$, consider $K$ independent random variables $\Lambda_i$ with values in $((\mathbf X^l)^{M_l})\times\M_{M_l}(\hr^{\otimes l}))$
which are distributed according to 
$\mu_l$.\\
Define, for each $s^l\in\bS^l$, the function $p_{s^l}:((\mathbf X^l)^{M_l})\times\M_{M_l}(\hr^{\otimes l}))\rightarrow[0,1]$,\\ 
$(x^l_1,\ldots,x^l_{M_l},D_1^l,\ldots,D_{M_l}^l)\mapsto\frac{1}{M_l}\sum_{i=1}^{M_l}\tr(A_{s^l}(x^n_i)D_i^l)$.\\
We get, by application of Markovs inequality, for every $s^l\in\bS^l$:
\begin{eqnarray}
 \mathbb P(1-\frac{1}{K}\sum_{j=1}^{K}p_{s^l}(\Lambda_j)\geq\eps/2)&=& \mathbb P(2^{ K- \sum_{j=1}^Kp_{s^l}(\Lambda_j)}\geq2^{K\eps/2})\\
&\leq&2^{-K\eps/2}\mathbb E(2^{ (K-\sum_{j=1}^Kp_{s^l}(\Lambda_j))}).
\end{eqnarray}
The $\Lambda_i$ are independent and it holds $2^{t}\leq1+t$ for every $t\in[0,1]$ as well as $\log(1+\eps_l)\leq2\eps_l$ and so we get
\begin{eqnarray}
 \mathbb P(1-\frac{1}{K}\sum_{j=1}^{K}p_{s^l}(\Lambda_j)\geq\eps/2)&\leq& 2^{-K\eps/2}\mathbb E(2^{K-\sum_{j=1}^Kp_{s^l}(\Lambda_j)})\\
&=&2^{-K\eps/2}\mathbb E(2^{\sum_{j=1}^K(1-p_{s^l}(\Lambda_j))})\\
&=&2^{-K\eps/2}\mathbb E(2^{(1-p_{s^l}(\Lambda_1))})^K\\
&\leq&2^{-K\eps/2}\mathbb E(1+(1-p_{s^l}(\Lambda_1)))^K\\
&\leq&2^{-K\eps/2}(1+\eps_l)^K\\
&\leq&2^{-K\eps/2}2^{K\eps/4}\\
&=&2^{-K\eps/4}.
\end{eqnarray}
Therefore,
\begin{equation}
  \mathbb P(\frac{1}{K}\sum_{j=1}^{K}p_{s^l}(\Lambda_j)\geq1-\eps/2)\geq1-|\bS|^l2^{-K\eps/4}.
\end{equation}
By assumption, $2\log|\bS|\leq l^{(n-m-1)}$ and thus the above probability is larger than zero, so there exists a realization 
$\Lambda_1,\ldots,\Lambda_{l^n}$ such that
\begin{equation}
\frac{1}{l^n}\sum_{i=1}^{l^n}\frac{1}{M_l}\tr(W_{s^l}(x^l_i)D_i^l)\geq1-\frac{1}{l^m}.
\end{equation}\qed
\end{proof}
Now we pass to the proof of Theorem \ref{ahlswede-dichotomy-for-avcqcs}. If $C_{A,r}(\A)=0$ or $C_{A,d}(\A)=0$ there is nothing to prove. So, let
$\overline C_{A,r}(\A)>0$ and 
$\overline C_{A,d}(\A)>0$. Then we know that, to every $l\in\nn$, there exists a deterministic code 
for $\A$ that, for sake of simplicity, is denoted by $(x^l_1,\ldots,x^l_{l^2},D_1,\ldots,D_{l^2})$, such that
\begin{align}
\min_{s^l\in\bS^l}\frac{1}{l^2}\sum_{i=1}^{l^2}\tr(A_{s^l}(x^l_i)D_i^l)\geq1-\eps_l
\end{align}
and $\eps_l\searrow0$. Also, by Lemma \ref{lemma:avcqc-random-achiev}, to every $\eps>0$ there is a sequence $(\mu_m)_{m\in\nn}$ of random codes for
transmission of messages over $\A$ using the average error 
probability criterion and an $m_0\in\nn$ such that
\begin{align}
&\liminf_{m\to\infty}\frac{1}{m}\log M_m\geq \overline C_{A,r}(\A)-\eps\\
\int\frac{1}{M_m}\sum_{j=1}^{M_m}\tr(A_{s^m}(x^m_j)&D_j^l)d\mu_m((x^m_1,\ldots,x^m_{M_m},D_1^l,\ldots,D_{M_m}^l))\geq1-2^{-mc}\ 
\end{align}
for all $m  \geq m_0$ with a suitably chosen (and possibly very small) $c>0$. This enables us to define the following sequence of codes:
Out of the random code, by application of Lemma \ref{random-code-reduction} and for a suitably chosen $m_1\geq m_0$ such that the preliminaries of
Lemma \ref{random-code-reduction} 
are fulfilled, we get for 
every $m\geq m_1$ a discrete random code supported only on the set $\{(y^m_{1,j},\ldots,y^m_{M_m,j},E_{1,j},\ldots,E_{M_m,j})\}_{j=1}^{m^2}$ such
that 
\begin{align}
\liminf_{m\to\infty}\frac{1}{l}\log M_m&\geq \overline C_{A,r}(\A)-\eps\\
\frac{1}{m^2}\sum_{j=1}^{m^2}\frac{1}{M_m}\sum_{i=1}^{M_m}\tr(A_{s^m}(y^m_{i,j})E_{i,j})&\geq1-\frac{1}{m}\qquad \forall m\geq m_1.
\end{align}
Now all we have to do is combine the two codes: For $l,m\in\nn$, define an $(l+m,\frac{1}{l^2M_m})$-deterministic code with the doubly-indexed message
set 
$\{i,j\}_{i=1,j=1}^{l^2,M_m}$ by the following sequence:
\begin{align}
((x_i^l,y_{ij}^m),D_i^l\otimes E_{ij})_{i=1,j=1}^{l^2,M_m}.
\end{align}
For the average success probability, by Lemma \ref{innerproduct-lemma} it then holds
\begin{align}
 \min_{(s^l,s^m)\in\bS^{l+m}}\frac{1}{l^2M_m}\sum_{i=1}^{l^2}\sum_{j=1}^{M_m}\tr(A_{(s^l,s^m)}((x^l_i,y^m_{ij}))D_i^l\otimes
E_{ij})\geq1-2\max\{\eps_l,\frac{1}{m}\}.
\end{align}
Now let there be sequences $(l_t)_{t\in\nn}$ and $(m_t)_{t\in\nn}$ such that $l_t=o(l)$ and $l_t+m_t=t$ f.a. $t\in\nn$. Define a sequence of
$(t,\frac{1}{l_t^2M_{m_t}})$-deterministic 
codes $(\hat x^t_1,\ldots,\hat x^t_{l_t^2M_{m_t}},\hat D_1,\ldots,\hat D_{l_t^2,M_{m_t}})$ for $\A$ by applying, for each $t\in\nn$, the above
described procedure with $m=m_t$ 
and $l=l_t$. Then 
\begin{align}
 \liminf_{t\to\infty}\frac{1}{t}\log l_t^2M_{m_t}\geq R\qquad \textup{and}\\
\lim_{t\to\infty}\min_{s^t\in\bS^t}\frac{1}{l_t^2M_{m_t}}\sum_{k=1}^{l_t^2M_{m_t}}\tr(A_{s^t}(\hat x_k^t)\hat D_k)=1.
\end{align}

\end{subsection}
\begin{subsection}{\label{subsec:maximal-and-zero-error}M-Symmetrizability}
In this section, we prove Theorem \ref{theorem:c-det=0-for-maximal-error}. 
\begin{proof}
We adapt the strategy of \cite{kiefer-wolfowitz}, that has already been successfully used in \cite{abbn}. Assume $\A$ is m-symmetrizable. Let
$l\in\nn$. Take any $a^l,b^l\in\mathbf X^l$. 
Then there exist corresponding probability distributions $p(\cdot|a_1),\ldots,p(\cdot|a_l),p(\cdot|b_1),\ldots,p(\cdot|b_l)\in\pr(\bS)$ such that the
probability distributions 
 $p(\cdot|a^l),p(\cdot|b^l)\in\mathfrak P(\bS^l)$ defined by $p(s^l|a^l):=\prod_{i=1}^lp(s_i|a_i)$, $p(s^l|b^l):=\prod_{i=1}^lp(s_i|b_i)$ satisfy
\begin{align}
 \sum_{s^l\in\bS^l}p(s^l|a^l)A_{s^l}(a^l)= \sum_{s^l\in\bS^l}p(s^l|b^l)A_{s^l}(b^l)
\end{align}
and thereby lead, for every two measurement operators $D_a,D_{b}\geq0$ satisfying $D_a+D_{b}\leq\eins_{\hr^{\otimes l}}$, to the following inequality:
\begin{align}
\sum_{s^l\in\bS^l}p(s^l|a^l)\tr(A_{s^l}(a^l)D_a)&=\sum_{s^l\in\bS^l}p(s^l|b^l)\tr(A_{s^l}(b^l)D_a)\\
&\leq \sum_{s^l\in\bS^l}p(s^l|b^l)\tr(A_{s^l}(b^l)(\mathbf1_{\hr^{\otimes l}}-D_b))\\
&=1-\sum_{s^l\in\bS^l}p(s^l|b^l)\tr(A_{s^l}(b^l)D_b).
\end{align}
Let a sequence of $(l,M_l)$ codes for message transmission over $\A$ using the maximal error probability criterion satisfying $M_l\geq2$ and 
$\min_{i\in[M_l]}\min_{s^l\in\bS^l}\tr(A_{s^l}(x^l_i)D^l_i)=1-\eps_l$ be given, where $\eps_l\searrow0$. Then from the above inequality we get
\begin{align}
 1-\eps_l\leq1-(1-\eps_l)\qquad &\Leftrightarrow\qquad\eps_l\geq1/2.
\end{align}
Therefore, $C_{A,d}(\A)=0$ has to hold.\\
Now, assume that $\A$ is not m-symmetrizable. Then there are $x,y\in\bX$ such that 
\begin{align}
\conv(\{\A_s(x)\}_{s\in\bS})\cap \conv(\{\A_s(y)\}_{s\in\bS})=\emptyset.
\end{align}
The rest of the proof is identical to that in \cite{abbn} with $\hat l$ set to one.\qed
\end{proof}

\end{subsection}
\begin{subsection}{\label{subsec:zero-error}Relation to the zero-error capacity}
A remarkable feature of classical arbitrarily varying channels is their connection to the zero-error capacity of (classical) d.m.c.s, which was
established by Ahlswede in \cite[Theorem 3]{ahlswede-note}.\\
We shall first give a reformulation of Ahlswede's original result and then consider two straightforward generalizations of it result, 
one for cq-channels, the other for 
quantum channels. In both cases it is shown, that no such straightforward generalization is possible.
\begin{paragraph}{Ahlswede's original result.}
Ahlswede's result can be formulated using the following notation. For two finite sets $\bA,\bB$, $C(\bA,\bB)$ stands for the set of channels from $\bA$ to $\bB$, i.e. each element 
of $W\in C(\bA,\bB)$ defines a set of output probability distributions $\{W(\cdot|a)\}_{a\in\bA}$. With slight abuse of notation, for each $D\subset \bB$ and $a\in\bA$, 
$W(D|a):=\sum_{b\in D}W(b|a)$. The (finite) set of extremal points of the (convex) set $C(\bA,\bB)$ will be written $E(\bA,\bB)$.\\
For two channels $W_1,W_2\in C(\bA,\bB)$, their product $W_1\otimes W_2\in C(\bA^2,\bB^2)$ 
is defined through $W_1\otimes W_2(b^2|a^2):=W_1(b_1|a_1)W_2(b_2|a_2)$. An arbitrarily varying channel (AVC) is, in this setting, defined through a set 
$\mathbb W=\{W_s\}_{s\in\bS}\subset C(\bA,\bB)$ (we assume $\bS$ and, hence, $|\mathbb W|$, to be finite). The different realizations of the channel are written
\begin{equation}
W_{s^l}:=W_{s_1}\otimes\ldots\otimes W_{s_l}\qquad (s^l\in\bS^l) 
\end{equation}
and, formally, the AVC $\mathbb W$ consists of the set $\{W_{s^l}\}_{s^l\in\bS^l,\ l\in\nn}$.\\
An $(l,M_l)$-code for the AVC $\mathbb W$ is given by a set 
$\{a^l_i\}_{i=1}^{M_l}\subset \bA^l$ called the 'codewords' and a set $\{D^l_i\}_{i=1}^{M_l}$ of subsets of $\bB^l$ called the 'decoding sets',
 that satisfies 
$D^l_i\cap D^l_j=\emptyset,\ i\neq j$.\\
A nonnegative number $R\in\rr$ is called an achievable maximal-error rate for the AVC $\mathbb W$, 
if there exists a sequence of $(l,M_l)$ codes for $\mathbb W$ such that both 
\begin{equation}
 \liminf_{l\to\infty}\frac{1}{l}\log M_l\geq R\qquad \mathrm{and}\qquad \lim_{l\to\infty}\min_{s^l\in\bS^l}\min_{1\leq i\leq M_l}W_{s^l}(D^l_i|x^l_i)=1.
\end{equation}
The (deterministic) maximal error capacity $C_{\mathrm{max}}(\mathbb W)$ of the AVC $\mathbb W$ is, as usually, defined as the supremum over all achievable maximal-error rates for $\mathbb W$.\\
Much stronger requirements concerning the quality of codes can be made. An $(l,M_l)$-code is said to have zero error for the AVC $\mathbb W$,
 if for all $1\leq i\leq M_l$ and 
$s^l\in\bS^l$ the equality $W_{s^l}(D^l_i|x^l_i)=1$ holds.\\
The zero error capacity $C_0(\mathbb W)$ of the AVC $\mathbb W$ is defined as 
\begin{equation}
C_0(\mathbb W):=\lim_{l\to\infty}\max\{\frac{1}{l}\log M_l:\exists\ (l,M_l)\mathrm{-code\ with\ zero\ error\ for\ }\mathbb W\}.
\end{equation}
The above definitions carry over to single channels $W\in C(\bA,\bB)$ by identifying $W$ with the set $\{W\}$.\\
In short form, the connection \cite[Theorem 3]{ahlswede-note} between the capacity of certain arbitrarily varying channels and the zero-error capacity of stationary memoryless 
channels can now be reformulated as follows:
\begin{theorem}\label{theorem:ahlswede-connection}
Let $W\in C(\bA,\bB)$ have a decomposition $W=\sum_{s\in\bS}q(s)W_s$, where $\{W_s\}_{s\in\bS}\subset E(\bA,\bB)$ and $q(s)>0\ \forall s\in\bS$. Then for the AVC 
$\mathbb W:=\{W_s\}_{s\in\bS}$:
\begin{equation}
C_0(W)=C_{\mathrm{max}}(\mathbb W).\label{eqn:ahlswede-connection-1}
\end{equation}
Conversely, for every AVC $\mathbb W=\{W_s\}_{s\in\bS}\subset E(\bA,\bB)$ and every $q\in\mathfrak P(\bS)$ with $q(s)>0\ \forall s\in\bS$, equation (\ref{eqn:ahlswede-connection-1}) 
holds for the channel $W:=\sum_{s\in\bS}q(s)W_s$.
\end{theorem}
\begin{remark}
Let us note at this point, that the original formulation of the theorem did not make reference to extremal points of the set of channels, but rather used the equivalent 
notion ''channels of $0-1$-type``.
\end{remark}
\begin{remark}
By choosing $W\in E(\bA,\bB)$, one gets the equality $C_0(W)=C_{\textrm{max}}(W)$. The quantity $C_{\textrm{max}}(W)$ being well-known and easily computable, it may seem that 
Theorem \ref{theorem:ahlswede-connection} solves Shannons's zero-error problem. This is not the case, as one can verify by looking at the famous pentagon channel that was introduced 
in \cite[Figure 2.]{shannon}. The pentagon channel is far from being extremal. That its zero-error capacity is positive \cite{shannon} is due to the fact that it is not a member of 
the relative interior $ri E(\bA,\bB)$.
\end{remark}
Recently, in \cite{abbn}, this connection was investigated with a focus on entanglement and strong subspace transmission over arbitrarily varying
quantum channels. The complete 
problem was left open, although partial results were obtained.
\end{paragraph}
\begin{paragraph}{A no-go result for cq-channels.}
We will show below that, even for message transmission over AVcqCs, there is (in general) no equality between the capacity $C_0(W)$ of a channel $W\in
CQ(\bX,\hr)$ and any AVcqC 
$\A=\{A_s\}_{s\in\bS}$ constructed by choosing the set $\{A_s\}_{s\in\bS}$ to be a subset of the set of extremal points of $CQ(\bX,\hr)$ such that
\begin{align}
 W=\sum_{s\in\bS}\lambda(s)A_s\label{eqn-21}
\end{align}
holds for a $\lambda\in\pr(\bS)$. Observe that the requirement that each $A_s$ ($s\in\bS$) be extremal in $CQ(\bX,\hr)$ is a natural analog of the
decomposition into channels of 
$0-1$-type that is used in the second part of \cite{ahlswede-note}.\\
A first hint why the above statement is true can be gained by looking at the method of proof used in \cite{ahlswede-note}, especially equation (22)
there. The fact that the decoding sets of a code for an arbitrarily 
varying channel as described in \cite{ahlswede-note} have to be mutually disjoint, together with the perfect distinguishability of different non-equal
outputs of the special 
channels that are used in the second part of this paper, is at the heart of the argumentation.\\
The following lemma shows why, in our case, it is impossible to make a step that is comparable to that from \cite[equation (21)]{ahlswede-note} to
\cite[equation (22)]{ahlswede-note}.
\begin{lemma}
Let $\A=\{A_s\}_{s\in\bS}$ be an AVcqC with $C_{A,d}(\A)>0$ and $0<R<C_{A,d}(\A)$. To every sequence of $(l,M_l)$ codes satisfying
$\liminf_{l\rightarrow\infty}\frac{1}{l}\log M_l\geq R$ and 
$\lim_{l\to\infty}\min_{i\in[M_l]}\min_{s^l\in\bS^l}\tr(A_{s^l}(x^l_i)D^l_i)=1$ there is another sequence of $(l,M_l)$ codes with modified decoding
operators $\tilde D^l_i$ 
such that
\begin{align}
1)&\qquad\liminf_{l\rightarrow\infty}\frac{1}{l}\log  M_l\geq R\\
2)&\qquad\lim_{l\to\infty}\min_{i\in[M_l]}\min_{s^l\in\bS^l}\tr(A_{s^l}(x^l_i)\tilde D^l_i)=1\\
3)&\qquad \forall\ i\in[M_l],\ l\in\nn,\ \ \tr(A_{s^l}(x^l_i)\tilde D_i^{l})<1
\end{align}
\end{lemma}
\begin{proof}
 Just use, for some $c>0$, the transformation $\tilde D^l_i:=(1-2^{-lc})D^l_i+2^{-lc}\frac{1}{M_l}(\eins_{\hr^{\otimes l}}-D_0^l)$.\qed
\end{proof}
After this preliminary statement, we give an explicit example that shows where the construction in equation (\ref{eqn-21}) must fail.
\begin{lemma}
Let $\bX=\{1,2\}$ and $\hr=\mathbb C^2$. Let $\{e_1,e_2\}$ be the standard basis of $\hr$ and $\psi_+:=\sqrt{1/2}(e_1+e_2)$. Define $W\in CQ(\bX,\hr)$
by $W(1)=|e_1\rangle\langle e_1|$ 
and $W(2)=|\psi_+\rangle\langle\psi_+|$. Then the following hold.
\begin{enumerate}
\item $W$ is extremal in $CQ(\bX,\hr)$
\item For every set $\{A_s\}_{s\in\bS}\subset CQ(\bX,\hr)$ and every $\lambda\in\pr(\bS)$ such that (\ref{eqn-21}) holds, $\{A_s\}_{s\in\bS}=\{W\}$.
\item $C_0(W)=0$, but $C_{A,d}(\{W\})>0$.
\end{enumerate}
\end{lemma}
\begin{proof}
1) Let, for an $x\in(0,1)$ and $W_1,W_2\in CQ(\bX,\hr)$,
\begin{align}
W=xW_1+(1-x)W_2.
\end{align}
Then, clearly,
\begin{align}
|e_1\rangle\langle e_1|=xW_1(1)+(1-x)W_2(1)\qquad\Longrightarrow\qquad W_1(1)=W_2(1)=W(1)
\end{align}
and
\begin{align}
|\psi_+\rangle\langle \psi_+|=xW_1(2)+(1-x)W_2(2)\qquad\Longrightarrow\qquad W_1(2)=W_2(2)=W(2),
\end{align}
so $W=W_1=W_2$.\\
2) is equivalent to 1).\\
3) It holds $\tr\{W(i)W(j)\}>1/2$ ($i,j\in\bX$). Let $l\in\nn$. Assume there are two codewords $a^l,b^l\in\bX^l$ and corresponding decoding operations
$C,D\geq0$, $C+D\leq\eins_{\mathbb C^2}^{\otimes l}$, such that 
\begin{align}
&\tr\{W^{\otimes l}(a^l)C\}=\tr\{W^{\otimes l}(b^l)D\}=1 \nonumber \\
(\Longrightarrow\ \ &\tr\{W^{\otimes l}(a^l)D\}=\tr\{W^{\otimes
l}(b^l)C\}=0).\label{eqn:example-2}
\end{align}
Then we may add a third operator $E:=\eins_{\mathbb C^2}^{\otimes l}-C-D$ and it holds that
\begin{align}
 \tr\{W^{\otimes l}(a^l)E\}=\tr\{W^{\otimes l}(b^l)E\}=0\label{eqn:example-1}.
\end{align}
From equations (\ref{eqn:example-1}) and (\ref{eqn:example-2}) we deduce the following:
\begin{align}\label{eqn:example-3}
&\sqrt{E}W^{\otimes l}(a^l)\sqrt{E}\ =\ \sqrt{E}W^{\otimes l}(b^l)\sqrt{E}\ \nonumber \\ 
=\ &\sqrt{D}W^{\otimes l}(a^l)\sqrt{D}\ =\ \sqrt{C}W^{\otimes l}(b^l)\sqrt{C}\ =\ 0.
\end{align}
With these preparations at hand, we are led to the following 
chain of inequalities:

\begin{align}
0&<\tr\{W^{\otimes l}(a^l)W^{\otimes l}(b^l)\}\\
 &=\tr\{W^{\otimes l}(C+D+E)(a^l)W^{\otimes l}(b^l)(C+D+E)\}\\
 &=\langle CW^{\otimes l}(a^l),W^{\otimes l}(b^l)C\rangle_{HS} 
   +\langle CW^{\otimes l}(a^l),W^{\otimes l}(b^l)D\rangle_{HS}& \nonumber\\
   & +\langle CW^{\otimes l}(a^l),W^{\otimes l}(b^l)E\rangle_{HS} 
   +\langle DW^{\otimes l}(a^l),W^{\otimes l}(b^l)C\rangle_{HS}& \nonumber \\
   & +\langle DW^{\otimes l}(a^l),W^{\otimes l}(b^l)D\rangle_{HS}
   +\langle DW^{\otimes l}(a^l),W^{\otimes l}(b^l)E\rangle_{HS}& \nonumber \\  
   & +\langle EW^{\otimes l}(a^l),W^{\otimes l}(b^l)C\rangle_{HS}
   +\langle EW^{\otimes l}(a^l),W^{\otimes l}(b^l)D\rangle_{HS}&  \nonumber \\
   & +\langle EW^{\otimes l}(a^l),W^{\otimes l}(b^l)E\rangle_{HS}& \\
 &=0,
\end{align}
as can be seen from a repeated application of the Cauchy-Schwarz-inequality to every single one of the above terms and use of equation
(\ref{eqn:example-3}). Thus, by 
contradiction, $C_0(W)=0$ has to hold.\\
Now assume that the AVcqC $\{W\}$ is m-symmetrizable. This is the case only if
\begin{align}
W(1)=W(2)
\end{align}
holds, which is clearly not the case. Thus, $C_{A,d}(\{W\})>0$.\qed
\end{proof}
\end{paragraph}
\begin{paragraph}{A no-go result for quantum channels.}
We now formulate a straightforward analogue of Theorem \ref{theorem:ahlswede-connection} for quantum channels. To this end, let us introduce some notation. We heavily rely 
on \cite{abbn}. The set of completely positive and trace-preserving maps from $\mathcal B(\hr)$ to $\mathcal B(\kr)$ (where both $\hr$ and $\kr$ are finite-dimensional) is denoted 
$\mathcal C(\hr,\kr)$. For a Hilbert space $\hr$, $S(\hr)$ denotes the set of vectors of unit lenght in it.\\
An arbitrarily varying quantum channel (AVQC) is defined by any set $\fri=\{\cn_s\}_{s\in\bS}\subset\mathcal C(\hr,\kr)$ and formally given by 
$\{\cn_{s^l}\}_{s^l\in\bS^l,l\in\nn}$, where
\begin{equation}
 \cn_{s^l}:=\cn_{s_1}\otimes\ldots\otimes\cn_{s_l}\qquad (s^l\in\bS^l).
\end{equation}
Let $\fri=\{\cn_s\}_{s\in\bS}$ be an AVQC. An $(l,k_l)-$\emph{strong subspace transmission code} for $\fri$ is a pair 
$(\cP^l,\crr^l)\in\mathcal C(\fr_l,\hr^{\otimes l})\times\mathcal C(\kr^{\otimes l},\fr_l')$, where $\fr_l,\ \fr_l'$ are Hilbert spaces and $\dim\fr_l=k_l$, $\fr_l\subset\fr_l'$. 
\begin{definition}\label{def:random-cap-strsub-trans}
A non-negative number $R$ is said to be an achievable strong subspace transmission rate for the AVQC $\fri=\{\cn_s  \}_{s\in\bS}$ if there is a sequence of $(l,k_l)-$strong subspace transmission codes such that
\begin{enumerate}
\item $\liminf_{l\rightarrow\infty}\frac{1}{l}\log k_l\geq R$ and
\item $\lim_{l\rightarrow\infty}\inf_{s^l\in\bS^l} \min_{\psi\in S(\fr_l)}\langle\psi,\crr^l\circ\cn_{s^l}\circ\cP^l(|\psi\rangle\langle\psi|)\psi\rangle=1$.
\end{enumerate}
The random strong subspace transmission capacity $\A_{\textup{s,random}}(\fri)$ of $\fri$ is defined by
\begin{eqnarray}\A_{\textup{s,det}}(\fri):=\sup\left\{R:\begin{array}{l} R \textrm{ is an achievable strong subspace}\\ \textrm{transmission rate for}\ \fri\end{array}\right \}.\end{eqnarray}
\end{definition}
Self-evidently, we will also need a notion of zero-error capacity:
\begin{definition}
An $(l,k)$ zero-error quantum code (QC for short) $(\fr,\cP,\crr)$ for $\cn\in \mathcal{C}(\hr,\kr) $ consists of a Hilbert space $\fr$, $\cP \in \mathcal{C}(\fr,\hr^{\otimes l})$, $\crr\in \mathcal{C}(\kr^{\otimes l},\fr)$ with $\dim \fr=k$ such that
\begin{equation}\label{def-q-0-code}
  \min_{x\in\fr, ||x||=1}\langle x, \crr\circ \cn^{\otimes l}\circ \cP(|x\rangle\langle x|) x\rangle =1.
\end{equation}
The zero-error quantum capacity $Q_0(\cn)$ of $\cn\in\mathcal{C}(\hr,\kr)$ is now defined by
\begin{equation}\label{0-error-q-capacity}
  Q_0(\cn):=\lim_{l\to\infty}\frac{1}{l}\log \max\{\dim \fr : \exists (l,k) \textrm{ zero-error QC for }\cn \}.
\end{equation}
\end{definition}
\begin{conjecture}\label{conjecture:quantum-ahlswede-connection}
Let $\cn\in \mathcal C(\hr,\kr)$ have a decomposition $\cn=\sum_{s\in\bS}q(s)\cn_s$, where each $\cn_s$ is extremal in $\mathcal C(\hr,\kr)$ and $q(s)>0\ \forall s\in\bS$. Then 
for the AVQC $\fri:=\{\cn_s\}_{s\in\bS}$:
\begin{equation}
Q_0(\cn)=\A_{\mathrm{s,det}}(\fri).\label{eqn:quantum-ahlswede-connection-1}
\end{equation}
Conversely, for every AVQC $\fri=\{\cn_s\}_{s\in\bS}$ with $\cn_s$ being extremal for every $s\in\bS$ and every $q\in\mathfrak P(\bS)$ with $q(s)>0\ \forall s\in\bS$, equation 
(\ref{eqn:quantum-ahlswede-connection-1}) holds for the channel $\cn:=\sum_{s\in\bS}q(s)\cn_s$.
\end{conjecture}
\begin{remark}\label{remark:quantum-ahlswede-connection} One could formulate weaker conjectures than the one above. A crucial property of extremal classical channels that 
was used in the proof of Theorem \ref{theorem:ahlswede-connection} was that $W_{s^l}(\cdot|x^l_i)$ is a dirac-measure for every codeword $x^l_i$, if only $\{W_{s^l}\}_{s\in\bS}\subset E(\bA,\bB)$.\\
This property gets lost for the extremal points of $\mathcal C(\hr,\kr)$ (see the channels that are used in the proof of Theorem \ref{theorem:conjecture-wrong}), but could be regained by restriction 
to channels consisting of only one single Kraus operator.
\end{remark}
This conjecture leads us to the following theorem:
\begin{theorem}\label{theorem:conjecture-wrong}
Conjecture \ref{conjecture:quantum-ahlswede-connection} is wrong.
\end{theorem}
\begin{remark}
As indicated in Remark \ref{remark:quantum-ahlswede-connection}, there could still be interesting connections between (for example) the deterministic strong subspace 
transmission capacity of AVQCs and the zero-error entanglement transmission of stationary memoryless quantum channels.
\end{remark}
\begin{proof}
 Let $\hr=\kr=\mathbb C^2$. Let $\{e_0,e_1\}$ be the standard basis of $\mathbb C^2$. Consider, for a fixed but arbitrary $x\in[0,1]$ the channel $\cn_x\in\mathcal C(\hr,\kr)$ 
defined by Kraus operators $A_1:=\sqrt{1-x^2}|e_0\rangle\langle e_1|$ and $A_2:=|e_0\rangle\langle e_0|+x|e_1\rangle\langle e_1|$. As was shown in \cite{wolf-cirac}, this 
channel is extremal in $\mathcal C(\hr,\kr)$. It is also readily seen from the definition of Kraus operators, that it approximates the identity channel 
$id_{\mathbb C^2}\in\mathcal C(\hr,\kr)$:
\begin{equation}\label{equation:ahlswede-connection-3}
 \lim_{x\to1}\|\cn_x-id_{\mathbb C^2}\|_\lozenge=0.
\end{equation}
Now, on the one hand, $\cn_x$ being extremal implies $\textrm{span}(\{A^*_iA_j\}_{i,j=1}^2)=M(\mathbb C^2)$ for all $x\in[0,1)$ (where $M(\mathbb C^2)$ denotes the set of complex $2\times2$ matrices) 
by \cite[Theorem 5]{choi}. This carries over to the channels $\cn_x^{\otimes l}$ for every $l\in\nn$: Let the Kraus operators of $\cn_x^{\otimes l}$ be denoted 
$\{A_{i^l}\}_{i^l\in\{1,2\}^l}$, then
\begin{equation}\label{equation:ahlswede-connection-2}
 \mathrm{span}(\{A_{i^l}^*A_{j^l}\}_{i^l,j^l\in\{1,2\}^l})=\{M:M\mathrm{\ is\ complex\ }2^l\times2^l\mathrm{-matrix}\}.
\end{equation}
On the other hand, it was observed e.g. in \cite{duan-severini-winter}, that for two pure states 
$|\phi\rangle\langle\phi|,|\psi\rangle\langle\psi|\in\cs((\mathbb C^2)^{\otimes l})$, the subspace spanned by them can be transmitted with zero error if and only if 
\begin{equation}
 |\psi\rangle\langle\phi|\perp\mathrm{span}(\{A_{i^l}^*A_{j^l}\}_{i^l,j^l\in\{1,2\}^l}).
\end{equation}
This is in obvious contradiction to equation (\ref{equation:ahlswede-connection-2}), therefore $Q_0(\cn_x)=0\ \forall x\in[0,1)$.\\
On the other hand, from equation (\ref{equation:ahlswede-connection-3}) and continuity of $\A_{\mathrm{s,det}}(\cdot)$ in the specifying channel set (\cite{abbn}, though indeed only the continuity 
results of \cite{leung-smith} that were also crucial in the development of corresponding statements in \cite{abbn} are really needed here) we see that there is an $X\in[0,1)$ 
such that for all $x\geq X$ we have $\A_{\mathrm{s,det}}(\{\cn_x\})>0$. Letting $x=X$ we obtain $Q_0(\cn_X)=0$ and $\A_{\mathrm{s,det}}(\{\cn_X\})>0$, so $Q_0(\cn_X)\neq \A_{\mathrm{s,det}}(\cn_X)$ 
in contradiction to the statement of the conjecture. \qed
\end{proof}
\end{paragraph}
\end{subsection}
\end{section}

\section*{Acknowledgments}
This work was supported by the DFG via grant BO 1734/20-1 (I.B, H.B.) and by the BMBF via grant 01BQ1050 (I.B., H.B., J.N.).  

\end{document}